\documentclass[a4paper,11pt]{article}
\usepackage{amsthm}
\usepackage{amssymb}
\usepackage{amsmath}
\usepackage{epsfig}
\usepackage{colonequals}
\usepackage{enumerate}
\usepackage{hyperref}

\newcommand{\Aa}{{\mathcal A}}
\newcommand{\GG}{{\mathcal G}}
\newcommand{\LL}{{\mathcal L}}

\newcommand{\Ss}{{\mathcal S}}

\newcommand{\PP}{{\mathcal P}}

\newcommand{\poly}{\text{poly}}
\newcommand{\tw}{\text{tw}}

\newtheorem{theorem}{Theorem}[section]
\newtheorem{definition}{Definition}
\newtheorem{corollary}[theorem]{Corollary}
\newtheorem{lemma}[theorem]{Lemma}
\newtheorem{observation}[theorem]{Observation}

\begin{document}
\title{Thin graph classes and polynomial-time approximation schemes}
\author{%
     Zden\v{e}k Dvo\v{r}\'ak\thanks{Computer Science Institute (CSI) of Charles University,
           Malostransk{\'e} n{\'a}m{\v e}st{\'\i} 25, 118 00 Prague, 
           Czech Republic. E-mail: \protect\href{mailto:rakdver@iuuk.mff.cuni.cz}{\protect\nolinkurl{rakdver@iuuk.mff.cuni.cz}}.
           Supported by (FP7/2007-2013)/ERC Consolidator grant LBCAD no. 616787.}
}
\date{\today}
\maketitle
\begin{abstract}
Baker~\cite{baker1994approximation} devised a powerful technique to obtain approximation schemes for various
problems restricted to planar graphs.  Her technique can be directly extended to various other graph classes,
among the most general ones the graphs avoiding a fixed apex graph as a minor.  Further generalizations (e.g.,
to all proper minor closed graph classes) are known, but they use a combination of techniques and usually focus
on somewhat restricted classes of problems.  We present a new type of graph decompositions (thin systems of overlays)
generalizing Baker's technique and leading to straightforward polynomial-time approximation schemes.
We also show that many graph classes (all proper minor-closed classes, and all subgraph-closed classes with bounded maximum
degree and strongly sublinear separators) admit such decompositions.
\end{abstract}

\section{Introduction}

Baker~\cite{baker1994approximation} devised a technique that leads to polynomial-time approximation schemes for
a range of problems (including maximum independent set, minimum dominating set, largest $H$-matching, minimum vertex cover,
and many others) when restricted to planar graphs.  In modern terms, the technique is based on the fact that
if $(V_1,\ldots, V_d)$ is a partition of vertices of a connected planar graph $G$ according to their distance from
an arbitrarily chosen vertex, then for any positive integer $s$ and $0\le i\le d-s+1$, the treewidth of $G[V_{i+1}\cup V_{i+2}\cup \ldots\cup V_{i+s-1}]$
is bounded by a function of $s$ (more specifically, it is at most $3s$), and hence many natural problems can be exactly solved
for these subgraphs in linear time~\cite{courcelle}.  As an example, the exact sizes of largest independent
sets in graphs $G-\bigcup_{i:i\bmod s=r} V_i$ for $r=0,1,\ldots,s-1$ can be found in linear time, and it is easy to see that one
of these sets has size at least $(1-1/s)$ times the size of the largest independent set in $G$.

It is natural to ask whether these algorithms can be extended to larger classes $\GG$ of graphs.
A \emph{layering} of a graph $G$ is a sequence $(V_1,\ldots, V_d)$, where $\{V_1,\ldots,V_d\}$
is a partition of $V(G)$ and $G$ contains no edges between sets $V_i$ and $V_j$ such that $|i-j|\ge 2$.
The sets $V_1$, \ldots, $V_d$ are the \emph{layers} of the layering.
Thus, assuming that $V_1$ intersects every component of $G$, the most natural layering is obtained by assigning to $V_i$
the vertices at distance exactly $i-1$ from $V_1$.
Let us say that a class $\GG$ has \emph{bounded treewidth layerings} if for some function $f$,
each graph $G\in\GG$ has a layering in which the union of any $s$ consecutive layers induces a subgraph of treewidth at most $f(s)$.
Baker's technique directly extends to all such graph classes (assuming that a suitable layering can be found
in polynomial time).  A natural obstruction for the existence of bounded treewidth layerings is as follows:
let $G_n$ be the graph obtained from the $n\times n$ grid by adding a universal vertex adjacent to all other vertices.
Then each layering of $G_n$ has at most three layers, and $\tw(G_n)>n$.  Hence, if a class $\GG$ has bounded treewidth layerings,
then it contains only finitely many of the graphs $\{G_n:n\ge 1\}$.  Conversely, Eppstein~\cite{eppstein00} proved that
proper minor-closed classes that avoid some graph $G_n$ have bounded treewidth layerings.

As another obstruction, let $U_n$ denote the graph obtained from an $n\times n\times n$ grid by adding all diagonals to its unit
subcubes.  Although the graphs $U_n$ have bounded maximum degree and very simple structure, 
Dvo\v{r}\'ak et al.~\cite{gridtw} proved that for every integer $k$, there exists $n_0$ such that for all $n\ge n_0$,
the vertex set of $U_n$ cannot be partitioned into two parts both inducing subgraphs of treewidth at most $k$.
This prevents existence of bounded treewidth layerings of $U_n$ (otherwise the partition of the graph into odd and even numbered layers
would give a contradiction).

In this paper, we introduce the concept of \emph{thin systems of overlays} which generalizes the notion of bounded treewidth layerings,
working around both mentioned obstructions.  In particular, we show that all proper minor-closed classes (Corollary~\ref{cor-minor})
as well as all subgraph-closed classes with bounded maximum degree and strongly sublinear separators (Corollary~\ref{cor-sublin}) admit such thin systems of overlays.
We also show that thin system of overlays can be used to design simple polynomial-time approximation schemes for many problems:
the previously mentioned maximum independent set and minimum dominating set and their distance variants,
largest $H$-matching, and minimum vertex cover, as well as any monotone maximization problems
expressible in a restricted fragment of Monadic Second Order logic with constant distance predicates
(see Theorem~\ref{thm-indep} for a precise formulation) which includes problems such as maximum $c$-colorable induced
subgraph for any constant number of colors $c$.  Let us remark that in all the cases, the approximation ratio only affects the multiplicative
constant of the time complexity of the algorithm (and not the degree of the polynomial), and in particular the results
also imply fixed-parameter tractability of the problems when parameterized by the order of the optimal solution.

The paper is organized as follows: In Section~\ref{sec-thinover} we introduce the notion of thin systems of overlays
and study its basic properties.  In Section~\ref{sec-approx} we present the polynomial-time approximation schemes obtained
using this notion.  Finally, in Section~\ref{sec-classes}, we give a number of operations on graph classes that
preserve the existence of thin systems of overlays, and as a corollary we show that all proper minor-closed classes as well as
all subgraph-closed classes with bounded maximum degree and strongly sublinear separators admit them.

\subsection{Related results}

As we already mentioned in the introduction, Eppstein~\cite{eppstein00} and Demaine and Hajiaghayi~\cite{demaine2004equivalence}
showed that Baker's technique generalizes
to all proper minor-closed classes that do not contain all apex graphs.  Going beyond the apex graph boundary, Grohe~\cite{grohe2003local},
Demaine et al.~\cite{demaine2005algorithmic} and Dawar et al.~\cite{dawar2006approximation} generalized
the approximation algorithms to all proper minor-closed classes using the tree decomposition from the minor structure theorem.
Baker's technique also applies to other geometrically defined graph classes, such as unit disk graphs~\cite{hunt1998nc}
and graphs embedded with a bounded number of crossings on each edge~\cite{grigoriev2007algorithms}.

Another powerful and largely independent approach to approximation algorithms in proper minor-closed graph classes
is through bidimensionality theory, bounding the treewidth of the graph in the terms of the size of the optimal solution
and exploiting the arising bounded-size balanced separators to obtain approximate solutions.
Demaine and Hajiaghayi~\cite{demaine2005bidimensionality} and Fomin et al.~\cite{fomin2011bidimensionality}
use this approach to construct polynomial-time approximation schemes on all proper minor-closed classes
for many minor-monotone problems (e.g., minimum vertex cover) and on apex-minor-free classes for contraction-monotone
problems (e.g., minimum dominating set).

Suppose that $\GG$ is a class of graphs with bounded treewidth layerings and $k$ is a positive integer.
As a particular corollary of Baker's technique, every graph $G\in\GG$ contains
pairwise disjoint sets $X_1,\ldots,X_k\subseteq V(G)$ such that removal of any of the sets results in a graph
whose treewidth is bounded by a constant (depending on the class of graphs and on $k$, but not on $G$ itself).
In terms introduced by Dvo\v{r}\'ak~\cite{twd}, such classes are \emph{treewidth-fragile}.  As shown by DeVos et al.~\cite{devospart}
and algorithmically by Demaine et al.~\cite{demaine2005algorithmic}, every proper minor-closed class is treewidth-fragile.
Treewidth-fragility is by itself sufficient to obtain polynomial-time approximation schemes for some graph parameters
such as the independence number or the size of the largest $H$-matching, but fails for others (minimum dominating set,
distance constrained versions of the independence number).  Demaine et al.~\cite{contrpart} give an edge-contraction
variation of this concept, which has similar algorithmic applications.

Some rather simple graph classes such as 3-dimensional
grids with diagonals in unit subcubes are not treewidth-fragile~\cite{gridtw}.  To overcome this obstruction,
Dvo\v{r}\'ak~\cite{twd} introduced the fractional version of treewidth-fragility which is still sufficient to obtain
the polynomial-time approximation schemes, and showed that all subgraph-closed graph classes with bounded maximum degree
and strongly sublinear separators are fractionally treewidth-fragile.  The notion of thin systems of overlays
can be seen as a further generalization that is able to handle approximation of wider class of problems such as minimum dominating set
and distance-$r$ independent set.  Let us remark that every class that admits thin systems of overlays is
fractionally treewidth-fragile.

A very different approach is taken by Cabello and Gajser~\cite{cabello2015simple} for proper minor-closed classes and more generally by
Har-Peled and Quanrud~\cite{har2015approximation} for classes of graphs with polynomial expansion (which by the
result of Dvo\v{r}\'ak and Norin~\cite{dvorak2016strongly} is equivalent to having strongly sublinear separators),
who showed that trivial local search algorithm gives polynomial-time approximation scheme for problems such as maximum
independent and minimum dominating set.  It is plausible (but not entirely clear) that this technique can be extended to
other first-order definable properties.

In conclusion, our approach seems to be the first one to give a graph decomposition notion generalizing Baker's bounded treewidth
layerings that is strong enough to give simple approximation schemes for most problems approachable via Baker's
technique, and at the same time can be applied to all proper minor-closed classes as well as many graph classes that are not minor-closed.

Let us remark that the approximation factor does not affect the exponent in the complexity of the obtained approximation schemes,
and consequently we also obtain fixed-parameter tractability when parameterized by the order of the optimum solution.
However, fixed-parameter tractability of the considered problems is generally known for much larger graph classes, see e.g.~\cite{grohe2014deciding}.

\subsection{Limitations of our technique}

Most of the algorithmic problems we consider, and in particular maximum independent set and minimum dominating set,
are APX-hard even when restricted to graphs with bounded maximum degree~\cite{papadimitriou1988optimization}.
Hence, we cannot hope to generalize our results to any class of graphs that includes say all 3-regular graphs,
and in particular we cannot consider topological minor-closed classes instead of minor-closed classes.

All the classes we consider have strongly sublinear separators; and it is easy to see that graphs from every class
that admits thin systems of overlays must have sublinear separators.  This is rather natural in the context of the
previous paragraph and the theory of bounded expansion: it is known that a subgraph-closed class has polynomial
expansion if and only if it has strongly sublinear separators~\cite{dvorak2016strongly}, and that subexponential
expansion implies existence of sublinear separators~\cite{grad2}, while the class of 3-regular graphs has
exponential expansion.  It would be interesting to see whether this intuition can be made precise and for example show
that the maximum independent set problem is APX-hard on any subgraph-closed class of graphs that does not have
sublinear separators (or at least has exponential expansion).

We show that subgraph-closed graph classes with strongly sublinear separators admit thin systems of overlays,
under the additional assumption that their maximum degree is bounded.  It seems plausible that this additional
assumption could be dropped, leading to a much more general result strengthening majority of the algorithms
mentioned in the previous subsection.  Nevertheless, even the simpler question of fractional treewidth-fragility
of such graph classes raised in~\cite{twd} is open.

Finally, let us remark that by its nature as a generalization of Baker's technique,
our technique is ill-suited for dealing with global problems, where parts of the solution interact over unlimited distances.
For the same reason, the approximation algorithms we design rely on at least a limited form
of monotonicity.  Thus, we do not obtain approximations for fundamentally non-monotone problems such as
minimum independent dominating set or minimum connected dominating set; the bidimensionality theory approach
is more suitable in these cases.

\section{Thin systems of overlays}\label{sec-thinover}

Before we give the formal definition of thin systems of overlays, let us start with a motivating example.
Given an input graph $G$, the \emph{minimum $r$-dominating set} problem (where $r$ is a fixed positive integer)
asks for a minimum-size set $X\subseteq V(G)$ such that each vertex of $G$ is at distance at most $r$ from $X$.
Suppose we want to apply Baker's technique to devise a polynomial-time approximation scheme for this problem
when restricted to connected \emph{planar} graphs.
Let us say that a set $Y\subseteq V(G)$ \emph{$r$-dominates} a set $Z\subseteq V(G)$ if each vertex of $Z$ is at distance
at most $r$ from $Y$.

Let $k$ be a positive integer.  Let us start with a layering $(V_1,\ldots,V_d)$ according to the distance from a vertex of $G$,
and for any integer $i$, let $G_i=G[V_{i-r}\cup \ldots\cup V_{i+kr+r-1}]$ and $Z_i=V_i\cup \ldots\cup V_{i+kr-1}$
(where $V_j=\emptyset$ if $j\le 0$ or $j>d$).   That is, $Z_i$ consists of $kr$ consecutive layers starting at the $i$-th one,
and $G_i$ is the subgraph of $G$ induced by $Z_i$ and all vertices
at distance at most $r$ from the set $Z_i$.  If $X$ is an $r$-dominating set in $G$, then the distance between any vertex of $Z_i$
and $X$ is at most $r$, and thus $X\cap V(G_i)$ $r$-dominates $Z_i$ in $G_i$.  Conversely,
let $Y_i\subseteq V(G_i)$ be a smallest set that $r$-dominates $Z_i$ in $G_i$,
which can be found in linear time since $G_i$ has bounded treewidth.
Let us fix any $a\in\{0,\ldots, kr-1\}$.  Then the sets $\{Z_i:i\bmod kr=a\}$ form a partition of $V(G)$,
and thus the set $Y^a=\bigcup_{i\bmod kr=a} Y_i$ is an $r$-dominating set in $G$.
Since the graphs $G_i$ such that $i \bmod kr=a$ only overlap in $r$ of each $rk$ consecutive layers,
it is easy to see that for some $a\in \{0,\ldots, kr-1\}$, the set $Y^a$ is at most $(1+1/k)$ times as large as a
minimum $r$-dominating set in $G$.  Thus, the algorithm can just return the smallest of the sets $Y_a$ for $a=0,\ldots, kr-1$.

This algorithm can be reformulated as follows.  For $a=0,\ldots,kr-1$, let $G^a$ be the graph obtained from graphs $G_i$ such that
$i\bmod kr=a$ as their disjoint union (first renaming the vertices that they have in common) and let $Z^a$ be analogously
obtained from sets $Z_i$ with $i\bmod kr=a$ (renaming the shared vertices in the same way).  Then $G^a$ has bounded
treewidth and $Y^a$ can be taken as the smallest set of vertices of $G^a$ that $r$-dominates $Z_a$, up to renaming.

Note that due to renaming, $G^a$ is not a subgraph of $G$; rather, there exists a surjective homomorphism $f$ from $G^a$ to $G$
that describes the renaming (let us recall that a \emph{homomorphism} is a function from $V(G^a)$ to $V(G)$ that maps edges to edges;
we only consider simple graphs without loops, and thus images of adjacent vertices must be distinct).
For each vertex $x\in Z^a$ and each walk $W'$ of length at most $r$ in $G$ starting in $f(x)$, there exists a walk $W$
of the same length in $G^a$ starting in $x$ such that $f(W)=W'$.  This is not the case for vertices in $V(G^a)\setminus Z^a$,
for which only smaller neighborhoods are represented in $G^a$. For each $x\in V(G^a)$, we let $\ell(x)$ denote the size
of the represented neighborhood, so we have $\ell(x)=r$ for $x\in V(G^a)$ and $\ell(\cdot)$ gets progressively smaller
in the layers away from $Z^a$.  Lastly, it is important that the number of renamed vertices is small compared to $|V(G)|$,
and that different vertices get renamed for different values of $a$.

With these observations in mind, let us proceed with formal definitions.

\begin{definition}\label{def-nonp-over}
Let $G$ and $H$ be graphs, and let $f:V(H)\to V(G)$ and $\ell:V(H)\to \mathbf{Z}_0^+$ be functions.
We say that $L=(H,f,\ell)$ is an \emph{overlay} of a graph $G$ if $f$ is a homomorphism from $H$ to $G$.
The overlay is \emph{walk-preserving} if for all $x\in V(H)$ with $\ell(x)\ge 1$ and for each neighbor $w$ of $f(x)$ in $G$,
there exists a neighbor $y$ of $x$ in $H$ such that $f(y)=w$ and $\ell(y)\ge\ell(x)-1$.
For a vertex $v\in V(G)$, the \emph{thickness} $\theta_L(v)$ of the overlay at $v$ is the number $|f^{-1}(v)|$ of vertices
of $H$ mapped to $v$.  The \emph{treewidth} of the overlay is equal to the treewidth of $H$.
\end{definition}

Since $f$ is a homomorphism from $H$ to $G$, each walk in $H$ corresponds to a walk in $G$.
The converse is not necessarily true, but the property of being walk-preserving implies the converse
at least for walks of bounded length.

\begin{lemma}\label{lemma-walk}
Let $L=(H,f,\ell)$ be an overlay of a graph $G$.  If $W$ is a walk in $H$ from a vertex $x$ to a vertex $y$,
then $f(W)$ is a walk in $G$ from $f(x)$ to $f(y)$.
Furthermore, if $L$ is walk-preserving and $W'$ is a walk in $G$ from $f(x)$ of length at most $\ell(x)$, then there exists a walk $W$ in
$H$ starting in $x$ such that $f(W)=W'$.
\end{lemma}
\begin{proof}
The first claim is obvious, since $f$ is a homomorphism.  For the second claim, let $W'=v_0v_1\ldots v_t$, where $f(x)=v_0$ and $t\le \ell(x)$.
Let $x_0=x$.  For $i=1, \ldots, t$, the assumption that $L$ is walk-preserving implies that there exists $x_i\in V(H)$ such that $x_{i-1}x_i\in E(H)$,
$f(x_i)=v_i$, and $\ell(x_i)\ge\ell(x_{i-1})-1\ge \ell(x)-i$.  Then we can set $W=x_0x_1\ldots x_t$.
\end{proof}

Based on the previous Lemma, the condition from the following definition guarantees that the neighborhood
of each vertex up to distance $r$ is represented in the overlay.

\begin{definition}\label{def-over}
Let $r$ be a positive integer, let $G$ be a graph and let $L=(H,f,\ell)$ be an overlay of $G$.
We say that $L$ is an \emph{$r$-neighborhood overlay} if $L$ is walk-preserving, $\ell(x)\in\{0,1,\ldots, r\}$ for all $x\in V(H)$, and
for all $v\in V(G)$, there exists $x\in f^{-1}(v)$ such that $\ell(x)=r$.
\end{definition}

Let us remark that an $r$-neighborhood overlay has positive thickness at each vertex.
For future use, let us note the following easy fact.

\begin{observation}\label{obs-sg}
If $G'$ is a subgraph of a graph $G$ and $L=(H,f,\ell)$ is an $r$-neighborhood overlay of $G$, then
an $r$-neighborhood overlay $L'$ of $G'$ with at most as large treewidth and thickness
can be constructed in time $O(|V(H)|+|E(H)|)$.
\end{observation}
\begin{proof}
To construct $L'=(H',f\restriction V(H'), \ell\restriction V(H'))$, let $H'$ be the graph obtained
from $H$ by removing all vertices $x$ with $f(x)\in V(G)\setminus V(G')$ and all edges $xy$ with
$f(x)f(y)\in E(G)\setminus E(G')$.
\end{proof}

In our motivating example, each component of the graphs $G^a$ was an induced subgraph of $G$.  This is not
necessarily the case for general overlays (and imposing this condition would restrict their applicability
significantly, as we will see below in Lemma~\ref{lemma-localtw}).
Nevertheless, it is convenient to keep track of this special case; we say that an overlay $L=(H,f,\ell)$ of a graph $G$
is \emph{subgraph-based} if the restriction of $f$ to each component $H'$ of $H$ is injective
and $f(H')$ is an induced subgraph of $G$.

Another variation on overlays will be needed to show that clique-sums of graphs from a class admitting thin
overlays form another class admitting thin overlays.
A vertex of a graph is \emph{simplicial} if its neighborhood induces a clique.
Suppose that all vertices of a set $S\subseteq V(G)$ are simplicial in $G$.
We say that an overlay $L=(H,f,\ell)$ of a graph $G$ is \emph{$S$-simpliciality-preserving} if for every vertex $v\in S$, all vertices
of $f^{-1}(v)$ are simplicial in $H$.
The graph $G^\star$ is obtained from $G$ by, for every non-empty-clique $K$ of $G$, adding a vertex $v_K$ adjacent to vertices in
$K$.  Note that if $G'$ is an (induced) subgraph of $G$, then $(G')^\star$ is an (induced) subgraph
of $G^\star$.  We define a \emph{$\star$-overlay of $G$} to be
a $(V(G^\star)\setminus V(G))$-simpliciality-preserving overlay of $G^\star$.

Let $\Aa_r(G)$ denote the class of all $r$-neighborhood overlays of $G$.
Let $\Ss_r(G)$ denote the class of all subgraph-based $r$-neighborhood overlays of $G$.
Finally, let $\bigstar_r(G)$ denote the class of all $r$-neighborhood $\star$-overlays of $G$.

\begin{definition}\label{def-system}
Consider any $\PP\in\{\Aa,\Ss,\bigstar\}$.
Let $r$ be a positive integer and let $G$ be a graph.
A \emph{system of $\PP_r$-overlays} of $G$ is a multiset $\LL$ of elements of $\PP_r(G)$.
Let $G'=G^\star$ if $\PP=\bigstar$ and $G'=G$ otherwise.
The \emph{thickness} $\theta_{\LL}(v)$ of a vertex $v\in V(G')$ is $\frac{1}{|\LL|}\sum_{L\in\LL}\theta_L(v)$, and
the thickness $\theta(\LL)$ of the system $\LL$ is the maximum of $\{\theta_{\LL}(v):v\in V(G')\}$.
The \emph{treewidth} $\tw(\LL)$ of the system is the maximum of the treewidths of its members.
\end{definition}

We say that the graph $G'$ from the previous definition is the \emph{graph overlaid by $\LL$}.
The definition of thickness of a system of overlays has a natural probabilistic interpretation which is sometimes convenient:
if the system $\LL$ has thickness $\theta$ and $L=(H,f,\ell)$ is chosen from $\LL$ uniformly at random, then for every
vertex $v$, the expected value of $\theta_L(v)$ is at most $\theta$.  For approximation algorithms,
an important consequence is that for any set $X\subseteq V(G)$, the expected size of $f^{-1}(X)$ is
at most $\theta |X|$, where typically we consider the case that $\theta$ is close to $1$.

As the approximation algorithm inspects all overlays in the system, we also need the system to be small;
for any fixed thickness, we need the size of the system for each graph $G$ in the class to be at most
polynomial in the number of vertices $n$ of $G$ (actually, for all graph classes we consider, the size will be either constant
or polynomial in $e^{(\log\log n)^2}$).  Furthermore, we of course need to be able to find such a system of overlays
efficiently.  With that in mind, we finally introduce the definition of graph classes that admit thin systems of overlays.

\begin{definition}\label{def-admit1}
Consider any $\PP\in\{\Aa,\Ss,\bigstar\}$.  Let $p_2$ be a non-decreasing positive polynomial in one variable,
let $q$ be a non-decreasing function, and let $r$, $k$, and $t$ be positive integers.
A class of graphs $\GG$ \emph{admits $(p_2,q)$-small $\PP_r$-overlays of thickness $1+1/k$ and treewidth $t$}
if for every graph $G\in \GG$ with $n$ vertices, there exists a system $\LL$ of $\PP_r$-overlays of $G$
of thickness at most $1+1/k$, treewidth at most $t$, and size at most $p_2(q(n))$.
\end{definition}

\begin{definition}\label{def-admit}
Consider any $\PP\in\{\Aa,\Ss,\bigstar\}$.  Let $q$ be a non-decreasing function and let $p$ be a non-decreasing positive polynomial.
Let $t$ be a function in two variables.
A class of graphs $\GG$ is \emph{$(\PP,q)$-thin}, with $t$ \emph{bounding the treewidth of its overlays}, if
for each positive integers $r$ and $k$, there exists a non-decreasing positive polynomial $p_2$
for which $\GG$ admits $(p_2,q)$-small $\PP_r$-overlays of thickness $1+1/k$ and treewidth $t(r,k)$.
We say that $\GG$ is \emph{$(\PP,p,q)$-thin} if there additionally exists a function $f$ in two variables
and an algorithm with time complexity $O(p(n)+n(q(n))^{f(r,k)})$ that for a graph $G\in \GG$ with $n$ vertices
returns the corresponding system of $\PP_r$-overlays.
\end{definition}

We say that the function $q$ is \emph{subpolynomial} if for every $\varepsilon>0$ and every positive integer $m$, $(q(n))^m=O(n^\varepsilon)$.
We will generally consider only classes that are $(\PP,q)$-thin for some subpolynomial $q$; we will not care about the exact
magnitude of the function $f(r,k)$, and we will write the complexity of the algorithm from the definition as
$O(p(n)+n\,\poly(q(n)))$, with the understanding that the degree of the polynomial
(and the multiplicative constants of the $O$-notation) may depend on $r$ and $k$.

Let us remark that in applications, we generally use only the case that $\PP=\Aa$.
As already mentioned, the notion of $\bigstar$-thinness is a technicality we introduce
to deal with clique-sums (although note also the application in Theorem~\ref{thm-svcov});
by Observation~\ref{obs-sg}, $\star$-overlays can be turned into normal overlays.
There is of course a concern that $G^\star$ could be much larger than $G$, affecting the time complexity of the operations.
However, this is not the case by the following lemma, which shows that for thin classes, $G^\star$ is larger only by a constant factor.

\begin{lemma}\label{lemma-nume}
Consider any $\PP\in\{\Aa,\Ss,\bigstar\}$.
If a class $\GG$ is $(\PP,q)$-thin with $t$ bounding the treewidth of its overlays and $c=4t(1,1)$, then
every graph $G\in\GG$ has maximum average degree at most $c$.  Consequently, $G$ is $c$-degenerate,
the largest clique in $G$ has size at most $c+1$, and $G$ contains at most $2^c|V(G)|$ non-empty cliques.
\end{lemma}
\begin{proof}
Consider a graph $G\in\GG$ and its subgraph $G'$, and let $\LL$ be a system of $\Aa_1$-overlays
of $G'$ of thickness at most $2$ and treewidth at most $t(1,1)$, which exists by Definition~\ref{def-admit}
and Observation~\ref{obs-sg}.
Let $L=(H,f,\ell)$ be chosen uniformly at random from $\LL$.  We have
$$\text{E}[|V(H)|]=\text{E}\left[\sum_{v\in V(G')}\theta_L(v)\right]\le 2|V(G')|$$
by the linearity of expectation.  Fix some $L\in \LL$ such that $|V(H)|\le 2|V(G')|$.

Since $L$ is walk-preserving, each vertex $x\in V(H)$ with $\ell(x)>0$ satisfies $\deg_H(x)\ge \deg_{G'}(f(x))$,
and we conclude that $$\sum_{x\in V(H)} \deg_H(x)\ge \sum_{v\in V(G')} \deg_{G'}(v),$$
and thus $|E(H)|\ge |E(G')|$.

Since $H$ has treewidth at most $t(1,1)$, it is $t(1,1)$-degenerate and its average degree is at most $2t(1,1)$.
Consequently, the average degree of $G'$ is
$$\frac{2|E(G')|}{|V(G')|}\le \frac{4|E(H)|}{|V(H)|}\le 4t(1,1).$$

Therefore, the maximum average degree of $G$ is at most $c$, and thus the graph $G$ is $c$-degenerate.  Consider the ordering
of vertices of $G$ such that each vertex has at most $c$ neighbors preceding it in the ordering.
If $K$ is a clique of $G$ and $v$ is the last vertex of $K$ in the ordering, then all other vertices of $K$
are among the preceding neighbors of $v$.  Hence, $|K|\le c+1$.  Furthermore, each vertex is the last vertex
of at most $2^c$ non-empty cliques, and thus the total number of non-empty cliques of $G$ is at most $2^c|V(G)|$.
\end{proof}

It follows that when we consider polynomial-time algorithms for graphs in a class $\GG$ admitting thin overlays of bounded treewidth,
it is sufficient to measure their time complexity with respect to the number of vertices,
as the number of edges is larger by at most a constant factor.

Let us remark that in the situation of Definition~\ref{def-admit1},
if $|\LL|=b$ and $\LL=\{L_i:1\le i\le b\}$
with $L_i=(H_i,f_i,\ell_i)$, then
$$\sum_{i=1}^b |V(H_i)|=\sum_{v\in V(G)}\sum_{i=1}^b \theta_{L_i} (v)\le nb\theta(\LL)=O(n\,\poly(q(n))).$$
Furthermore, since all the graphs $H_i$ have bounded treewidth, the number of their
edges is linear in the number of their vertices.  Hence linear-time manipulations with such systems of overlays
can be performed in time $O(n\,\poly(q(n)))$, and we will neglect mentioning the complexity of such manipulations afterwards.
In particular, these facts together with Observation~\ref{obs-sg} imply that every $(\bigstar,p,q)$-thin class is $(\Aa,p,q)$-thin.

\section{Approximation algorithms via thinness}\label{sec-approx}

Let us now show how thin systems of overlays can be used to design approximation schemes.
Firstly, we will consider problems such as the following ones:
\begin{itemize}
\item Maximum $r$-independent set, for a fixed positive integer $r$: Given a graph $G$, find a largest subset $X$ of vertices of $G$ such that the distance
between distinct elements of $X$ is at least $r$.
\item Maximum $F$-matching, for a fixed connected graph $F$: Given a graph $G$, find a largest subset $X$ of vertices of $G$ such that $G[X]$ can be partitioned into
vertex-disjoint copies of $F$.
\item Maximum $c$-colorable induced subgraph, for a fixed positive integer $c$: Given a graph $G$, find a largest subset $X$ of vertices of $G$ such that
the induced subgraph $G[X]$ can be partitioned into $c$ independent sets.
\item The $r$-distance variants of the previous two problems, for any fixed integer $r$.
\end{itemize}

All these problems can be expressed in uniform way as follows.
For a positive integer $r$, let $\text{MSO}_r$ denote Monadic Second Order logic with predicates $d_l(u,v)$ for $l=1,\ldots,r$,
where $d_l(u,v)$ is interpreted to be true if there exists a path of length at most $l$ between vertices $u$ and $v$.
A formula $\varphi$ in $\text{MSO}_r$ with one free vertex set variable $X$ is \emph{solution-restricted}
if all its quantifications are restricted to $X$, i.e., $\varphi$ only contains quantifiers of form
$(\exists/\forall v:v\in X)$, $(\exists/\forall V:V\subseteq X)$, $(\exists/\forall e=uv:u,v\in X)$, and
$(\exists/\forall E:u,v\in X\text{ for all }uv\in E)$).

Consider the problem of finding a largest set $X\subseteq V(G)$ in a graph $G$ satisfying some property $\pi$.
For a non-negative integer $s$, we say that the problem is \emph{$s$-near-monotone solution-restricted $\text{MSO}_r$-expressible} if
the following conditions are satisfied.
\begin{itemize}
\item For every set $X\subseteq V(G)$ that has the property $\pi$ and every set $Y\subseteq X$,
there exists a subset of $Y$ of size at least $|Y|-s|X\setminus Y|$ that has the property $\pi$.
\item There exists a solution-restricted formula $\varphi$ in $\text{MSO}_r$ with one free vertex set variable $X$
such that $G\models \varphi(X)$ if and only if $X$ has the property $\pi$.
\end{itemize}
As an example, consider the $r$-distance variant of maximum $F$-matching, for a positive integer $r$ and a connected graph $F$.
Here, $\pi$ is the property that there exists a set of edges $Z\subseteq E(G[X])$ such that the graph $(X,Z)$ is a disjoint union
of copies of $F$ and no two vertices belonging to different components of the graph have distance less than $r$ in $G$;
this property is clearly solution-restricted $\text{MSO}_r$-expressible.  Furthermore, if $X$ has this property $\pi$
and $Y\subseteq X$, then let $Y'$ be the subset of $Y$ obtained by removing all components of the graph $(X,Z)$ that intersect $X\setminus Y$.
Clearly, $Y'$ has property $\pi$ and $|Y'|\ge |Y|-(|V(F)|-1)|X\setminus Y|$, and thus $\pi$ is $(|V(F)|-1)$-near-monotone.

\begin{theorem}\label{thm-indep}
For integers $s\ge 0$ and $r\ge 1$, let $\pi$ be an $s$-near-monotone solution-restricted $\text{MSO}_r$-expressible property,
and let $\varepsilon>0$ be any real number.
If a class of graphs $\GG$ is $(\Aa,p,q)$-thin for some functions $p$ and $q$, then given
an $n$-vertex graph $G\in \GG$, a set $X\subseteq V(G)$ satisfying the property $\pi$
whose size is within $1-\varepsilon$ factor of the maximum one can be found in time $O(p(n)+n\,\poly(q(n)))$.
\end{theorem}
\begin{proof}
Let $\varphi$ be a solution-restricted $\text{MSO}_r$-formula expressing the property $\pi$.
Let $k$ be a positive integer such that $1/k\le \varepsilon/(s+1)$.
Let $\LL$ be a system of $\Aa_r$-overlays of $G$ of thickness at most $1+1/k$, treewidth at most $t=t(r,k)$ and size $\poly(q(n))$,
which can be found in time $O(p(n)+n\,\poly(q(n)))$.

Consider any overlay $L=(H,f,\ell)\in\LL$, and let $S_L=\{v\in V(G):\theta_L(v)=1\}$.
Note that $\ell(x)=r$ for all $x\in f^{-1}(S_L)$, and that $f$ is an isomorphism from $H[f^{-1}(S_L)]$ to $G[S_L]$.
Since $\varphi$ is solution-restricted, only the distance predicates in $\varphi$ take into account the parts
of the graph not contained in the considered set; hence, Lemma~\ref{lemma-walk} implies that for any $Z\subseteq S_L$,
we have $G\models \varphi(Z)$ if and only if $H\models\varphi(f^{-1}(Z))$.
Since $H$ has bounded treewidth, we can in time $O(|V(H)|)$ find a set $Z_L\subseteq S_L$ of largest size such that
$H\models\varphi(f^{-1}(Z_L))$, and thus $Z_L$ has property $\pi$.
The algorithm returns the largest set $Z_L$ over all $L\in\LL$.

Consider an optimal solution, i.e., a set $X_0\subseteq V(G)$ of maximum size satisfying the property $\pi$.
For any $L=(H,f,\ell)\in\LL$, let $Y_L=X_0\cap S_L$.  Since $\pi$ is $s$-near-monotone, there exists a set $Z'_L\subseteq Y_L$ of size at
least $|Y_L|-s|X_0\setminus Y_L|$ satisfying property $\pi$; since $Z'_L\subseteq S_L$, the choice of $Z_L$ implies
$$|Z_L|\ge |Z'_L|\ge |X_0|-(s+1)|X_0\setminus Y_L|=|X_0|-(s+1)|X_0\setminus S_L|.$$
Consider any vertex $v\in V(G)$.  Since $\theta_\LL(v)\le 1+1/k$, and $\theta_L(v)\ge 1$ for all overlays $L\in\LL$, and $\theta_L(v)\ge 2$
if $v\not\in S_L$, we conclude that there are at most $|\LL|/k$ overlays $L\in\LL$ such that $v\not\in S_L$.
Consequently,
$$\frac{1}{|\LL|}\sum_{L\in\LL}|X_0\setminus S_L|=\sum_{v\in X_0}\frac{1}{|\LL|}|\{L\in\LL:v\not\in S_L\}|\le \frac{|X_0|}{k},$$
and thus there exists an overlay $L\in\LL$ such that $|X_0\setminus S_L|\le |X_0|/k$.  For this overlay, we have $$|Z_L|\ge |X_0|-\frac{s+1}{k}|X_0|\ge (1-\varepsilon)|X_0|,$$
and thus the algorithm returns a set whose size is within $1-\varepsilon$ factor of the maximum one.
\end{proof}

A similar argument can be applied to deal with domination problems.

\begin{theorem}\label{thm-domin}
Let $r$ be a positive integer and $\varepsilon$ a positive real number.
If a class of graphs $\GG$ is $(\Aa,p,q)$-thin for some functions $p$ and $q$,
then given an $n$-vertex graph $G\in \GG$, an $r$-dominating set $X\subseteq V(G)$
whose size is within $1+\varepsilon$ factor of the minimum one can be found in time $O(p(n)+n\,\poly(q(n)))$.
\end{theorem}
\begin{proof}
Let $k$ be a positive integer such that $1/k\le \varepsilon$.
Let $\LL$ be a system of $\Aa_r$-overlays of $G$ of thickness at most $1+1/k$, treewidth at most $t=t(r,k)$ and size $\poly(q(n))$,
which can be found in time $O(p(n)+n\,\poly(q(n)))$.

Consider any overlay $L=(H,f,\ell)\in\LL$, and let $S_L=\{x\in V(H):\ell(x)=r\}$.
Since $L$ is an $r$-neighborhood overlay, we have $f(S_L)=V(G)$.
Let $Z_L$ be a minimum subset of $V(H)$ that $r$-dominates $S_L$ in $H$, which can be found in time $O(|V(H)|)$ since $H$ has bounded
treewidth.  By Lemma~\ref{lemma-walk}, $f(Z_L)$ is an $r$-dominating set in $G$.  The algorithm returns the smallest set $f(Z_L)$ over all $L\in\LL$.

Let $X_0\subseteq V(G)$ be a minimum $r$-dominating set of $G$.
For an overlay $L=(H,f,\ell)\in\LL$, let $Z'_L=f^{-1}(X_0)$. By Lemma~\ref{lemma-walk}, the set $Z'_L$ $r$-dominates $S_L$ in $H$,
and thus $|Z'_L|\ge |Z_L|$.  Furthermore, $|Z'_L|=\sum_{v\in X_0}\theta_L(v)$.
Since $\LL$ has thickness at most $1+1/k$, there exists an overlay $L\in\LL$ such that $|Z'_L|\le (1+1/k)|X_0|\le (1+\varepsilon)|X_0|$.
Consequently, $|f(Z_L)|\le |Z_L|\le |Z'_L|\le (1+\varepsilon)|X_0|$,
and thus the algorithm returns a set whose size is within $1+\varepsilon$ factor of the maximum one.
\end{proof}

The notion of $\bigstar$-thinness is mostly a technical tool to deal with clique-sums.  Nevertheless,
it also gives us a control over the cliques in the overlaid graph, which is useful in the
problem of finding a minimum vertex cover, or more generally in the following problem.
An \emph{$s$-clique cover} is a set of vertices of $G$ that intersects all cliques of size at least $s$.

\begin{theorem}\label{thm-svcov}
Let $s$ be a positive integer and $\varepsilon$ a positive real number.
If a class of graphs $\GG$ is $(\bigstar,p,q)$-thin for some functions $p$ and $q$, then given an $n$-vertex graph $G\in \GG$,
an $s$-clique cover whose size is within $1+\varepsilon$ factor of the minimum one can be found in time $O(p(n)+n\,\poly(q(n)))$.
\end{theorem}
\begin{proof}
Let $S$ be the set of vertices of $G^\star$ representing the cliques of $G$ of size $s$.
Let $k$ be a positive integer such that $1/k\le \varepsilon$.
Let $\LL$ be a system of $\bigstar_1$-overlays of $G$ of thickness at most $1+1/k$, treewidth at most $t=t(1,k)$ and size $\poly(q(n))$,
which can be found in time $O(p(n)+n\,\poly(q(n)))$.

Consider any overlay $L=(H,f,\ell)\in\LL$.  Let $S_L=\{x\in f^{-1}(S): \ell(x)=1\}$.
Since $L$ is $S$-simpliciality-preserving, the neighborhood of each vertex $x\in S_L$ induces a clique,
and since $\ell(x)=1$ and $L$ is walk-preserving, the size of this clique is exactly $s$.
Let $Z_L$ be a minimum subset of $V(H)$ intersecting the neighborhoods of all vertices of $S_L$, which can be found in time $O(|V(H)|)$ since $H$ has bounded
treewidth.  Since $L$ is a $\star$-overlay of $G$, $f(Z_L)$ is an $s$-clique cover of $G$.
The algorithm returns the smallest set $f(Z_L)$ over all $L\in\LL$.

Let $X_0\subseteq V(G)$ be a minimum $s$-clique cover of $G$.
For any overlay $L=(H,f,\ell)\in\LL$, let $Z'_L=f^{-1}(X_0)$, and observe that $Z'_L$ intersects
the neighborhoods of all vertices of $S_L$.  Hence, $|Z'_L|\ge |Z_L|$.
Furthermore, $|Z'_L|=\sum_{v\in X_0}\theta_L(v)$.
Since $\LL$ has thickness at most $1+1/k$, there exists an overlay $L\in\LL$ such that $|Z'_L|\le (1+1/k)|X_0|\le (1+\varepsilon)|X_0|$.
Consequently, $|f(Z_L)|\le |Z_L|\le |Z'_L|\le (1+\varepsilon)|X_0|$,
and thus the algorithm returns a set whose size is within $1+\varepsilon$ factor of the maximum one.
\end{proof}

Regarding the time complexity of the algorithms, for all the classes of graphs that we consider the function $q$ is subpolynomial,
and thus the time complexity turns out to be at most $O(p(n)+n^{1+\delta})$ for any fixed $\delta>0$, with only the multiplicative
constant depending on the desired precision of approximation.

\section{$\Ss$-thinness and $\bigstar$-thinness}

As we have seen at the beginning of Section~\ref{sec-thinover},
subgraph-based overlays are an interesting special case arising
in previously studied settings, and in many ways $\Ss$-thinness is easier to work with and might give stronger algorithmic consequences.
Nevertheless, as we will show momentarily, only classes with locally bounded treewidth can be $\Ss$-thin
(recall that a class $\GG$ has \emph{locally bounded treewidth} if for each positive integer $r$,
the treewidth of $r$-neighborhoods of vertices in graphs from $\GG$ is bounded by a constant depending only on $r$
and the class).  This restricts the applicability of the concept, and in particular shows that not all proper
minor-closed classes are $\Ss$-thin.

\begin{lemma}\label{lemma-localtw}
For any function $q$, if a class of graphs $\GG$ is $(\Ss,q)$-thin, then $\GG$ has locally bounded treewidth.
\end{lemma}
\begin{proof}
Let $t$ be the function bounding the treewidth of overlays of $\GG$.
Consider any positive integer $r$.
Every graph $G\in\GG$ has a system $\LL$ of $\Ss_r$ overlays of thickness at most $2$
and treewidth at most $t(r,1)$.  Consider any overlay $L=(H,f,\ell)\in\LL$ and any vertex $v\in V(G)$.  Since $L$ is an $r$-neighborhood overlay of $G$,
there exists a vertex $x\in V(H)$ such that $f(x)=v$ and $\ell(x)=r$.  Since $L$ is subgraph-based,
Lemma~\ref{lemma-walk} implies that the component $H'$ of $H$ with $x\in V(H')$ has a subgraph isomorphic to the subgraph of $G$ induced by vertices
at distance at most $r$ from $v$.  Hence, this subgraph must have treewidth at most $\tw(H')\le t(r,1)$.
As the choice of $G$, $v$ and $r$ was arbitrary, we conclude that $\GG$ has locally bounded treewidth.
\end{proof}

Trivially, $\Ss$-thinness implies $\Aa$-thinness. Let us also remark that $\Ss$-thinness implies $\bigstar$-thinness.

\begin{lemma}\label{lemma-sgbas-star}
For any functions $p$ and $q$, if a class of graphs $\GG$ is $(\Ss,p,q)$-thin, then $\GG$ is also $(\bigstar,p,q)$-thin.
\end{lemma}
\begin{proof}
Let the function $t$ bound the treewidth of overlays of $\GG$.  Consider any positive integers $k$ and $r$,
and a graph $G\in \GG$ with $n$ vertices.
Let $\LL$ be the system of $\Ss_r$-overlays of $G$ of thickness at most $1+1/k$, found in time $O(p(n)+n\,\poly(q(n)))$ according to Definition~\ref{def-admit}.
Consider any overlay $L=(H,f,\ell)\in \LL$.  Let $L_1=(H_1,f_1,\ell_1)$ be obtained from $L$ as follows.
The graph $H_1$ is a supergraph of $H$ with functions $f_1$ and $\ell_1$ matching $f$ and $\ell$ on $V(H)$.
For each non-empty clique $K$ of $G$ and each component $H'$ of $H$ such that $K\subseteq f(V(H'))$, we add a new vertex $x$ to $H_1$
with neighborhood $V(H')\cap f^{-1}(K)$.  We set $f_1(x)=v_K$, where $v_K$ is the vertex of $G^\star$ representing the clique $K$.  Also, we set
$\ell_1(x)=\max\{\ell(y):y\in V(H')\cap f^{-1}(K)\}$.

Let us verify that $L_1$ is an $r$-neighborhood $\star$-overlay of $G$.  By construction, it is a $(V(G^\star)\setminus V(G))$-simpliciality-preserving overlay of $G^\star$.
Let $x$ be any vertex of $H_1$ with $\ell_1(x)\ge 1$ and let $w$ be a neighbor of $f_1(x)$ in $G^\star$.  Note that at most one of $f_1(x)$ and $w$ belongs to $V(G^\star)\setminus V(G)$.
If $f_1(x),w\in V(G)$, then $x$ has a neighbor $y$
such that $f_1(y)=f(y)=w$ and $\ell_1(y)=\ell(y)\ge\ell(x)-1=\ell_1(x)-1$ since $L$ is a walk-preserving overlay of $G$.
If $w\in V(G^\star)\setminus V(G)$, then let $H'$ be the component of $H$ containing $x$ and let $K$ be the clique in $G$ such that $w=v_K$.
Note that $f_1(x)=f(x)\in K$, and since $\ell(x)=\ell_1(x)\ge 1$ and $L$ is a subgraph-based overlay, we conclude that $K\subseteq f(V(H'))$;
hence, a neighbor $y$ of $x$ such that $f_1(y)=w$ and $\ell_1(y)\ge\ell_1(x)$ was added to $H_1$ during its construction.
If $f_1(x)\in V(G^\star)\setminus V(G)$, then let $K$ be the clique in $G$ such that $f_1(x)=v_K$; we have $w\in K$,
and by the construction of $H_1$, $x$ has a unique neighbor $y$ in $H_1$ such that $f_1(y)=w$.  Let $H'$ be the component of $H$ containing $y$.
Since $L$ is a subgraph-based walk-preserving overlay, note that $\ell(f^{-1}(K)\cap V(H'))$ consists of at most two consecutive integers,
and thus $\ell_1(y)\ge \ell_1(x)-1$.  Consequently, the overlay $L_1$ is walk-preserving.
Consider any non-empty clique $K$ of $G$ and a vertex $z\in K$.  Since $L$ is an $r$-neighborhood overlay,
there exists a vertex $x\in V(H)$ with $f(x)=z$ and $\ell(x)=r$.  Let $H'$ be the component of $H$ containing $x$.
Since $L$ is subgraph-based and walk-preserving, we conclude that $K\subseteq f(V(H'))$, and thus during the construction of $H_1$,
we added a neighbor $w$ of $x$ such that $f_1(w)=v_K$ and $\ell_1(w)=r$.  Therefore, $L_1$ is an $r$-neighborhood overlay.

Consequently, $L_1\in \bigstar_r(G)$.  Note that $L_1$ has treewidth at most $\tw(L)+1$.  Furthermore, the system $\{L_1:L\in\LL\}$ has thickness at most $1+1/k$,
since the thickness at each vertex $v_K$ is bounded by the thicknesses of the vertices of $K$.
Since the described transformation can be performed in time $O(n\,\poly(q(n)))$, we conclude that $\GG$ is $(\bigstar,p,q)$-thin.
\end{proof}

\section{Thinness of graph classes}\label{sec-classes}

Finally, we show that many important graph classes are thin.  Let us start by establishing
some basic tools.

\subsection{Transformations of thin classes}
When considering systems of overlays for several graphs, it is sometimes convenient to have all of them of the same
size, which is possible by the following lemma.

\begin{lemma}\label{lemma-samesize}
Consider any $\PP\in\{\Aa,\Ss,\bigstar\}$.
Let $\GG$ be a $(\PP,p,q)$-thin class of graphs for some functions $p$ and $q$, with $t$ bounding the treewidth of its overlays.
Let $G_1$, \ldots, $G_m$ be graphs from $\GG$,
each with at most $n$ vertices, and let $N=\sum_{i=1}^m |V(G_i)|$.
For any positive integers $r$ and $k$, we can find systems $\LL_1$, \ldots, $\LL_m$ of $\PP_r$-overlays of graphs $G_1$, \ldots, $G_m$ of
thickness at most $1+1/k$ and treewidth at most $t(r,3k)$ such that $|\LL_1|=\ldots=|\LL_m|\le \poly(q(n))$
in time $O(p(N)+N\,\poly(q(n)))$.
\end{lemma}
\begin{proof}
For $1\le i\le m$, since the class $\GG$ is $(\PP,p,q)$-thin, we can find a system $\LL'_i$ of $\PP_r$-overlays of $G_i$ of thickness at most $1+\frac{1}{3k}$,
treewidth at most $t(r,3k)$, and order at most $\poly(q(n))$ in time $O(p(|V(G_i)|)+|V(G_i)|\poly(q(n)))$.
Let $a=\max\{|\LL_i|:1\le i\le m\}$.
For $i\in\{1,\ldots,m\}$, let $\LL_i$ be a multiset obtained from $\LL'_i$ by replicating its elements
so that $|\LL_i|=3ka\le \poly(q(n))$ and the number of times an element of $\LL'_i$ is replicated is either $c_i$ or $c_i+1$
for some positive integer $c_i$.  The total time needed to obtain these systems clearly is $O(p(N)+N\,\poly(q(n)))$,
hence we just need to show that they have thickness at most $1+1/k$.

Consider any $i\in\{1,\ldots, m\}$, and let $\LL'_i=\{L_1,\ldots,L_{b_i}\}$, where $b_i\le a$.  Note that $3ka=|\LL_i|\ge b_ic_i$.  If $v$ is any vertex of the graph overlaid by $\LL_i$, then
\begin{align*}
\theta_{\LL_i}(v)&\le\frac{c_i+1}{3ka}\sum_{j=1}^{b_i} \theta_{L_j}(v)\\
&=\frac{b_i(c_i+1)}{3ka}\theta_{\LL'_i}(v)\\
&\le \Bigl(1+\frac{b_i}{3ka}\Bigr)\theta_{\LL'_i}(v)\le \Bigl(1+\frac{1}{3k}\Bigr)\theta_{\LL'_i}(v)\\
&\le \Bigl(1+\frac{1}{3k}\Bigr)^2<1+\frac{1}{k},
\end{align*}
as required.
\end{proof}

For $1\le i\le m$, let $L_i=(H_i,f_i,\ell_i)$ be an overlay of a graph $G$.
We define $L_1\circ\ldots\circ L_m$ as the overlay $(H,f,\ell)$ of $G$, where $H$ is the disjoint union of
$H_1$, \ldots, $H_m$, and for each $i\in\{1,\ldots,m\}$ and each vertex $v\in V(G_i)$ we have
$f(v)=f_i(v)$ and $\ell(v)=\ell_i(v)$.

Let $\LL_1$, \ldots, $\LL_m$ be systems of overlays of a graph $G$ such that $|\LL_1|=\ldots=|\LL_m|=a$.
For $i\in\{1,\ldots,m\}$, let $\LL_i=\{L_i^1,\ldots,L_i^a\}$, with the elements listed in an arbitrary order. For $j\in\{1,\ldots,a\}$, let
$L^j=L_1^j\circ\ldots\circ L_m^j$, and we define $\LL_1\circ\ldots\circ\LL_m$ to denote the system of overlays $\{L^1,\ldots, L^a\}$.
We also use $\bigcirc_{i=1}^m\LL_i$ to denote $\LL_1\circ\ldots\circ\LL_m$ and other analogous $\sum$-like notation for the $\circ$ operation.

\begin{lemma}\label{lemma-comp}
Consider any $\PP\in\{\Aa,\Ss,\bigstar\}$.
Let $\GG$ be a $(\PP,p,q)$-thin class of graphs for some functions $p$ and $q$,
and let $\GG'$ be a class of graphs.  If for every $G\in\GG'$, each component of $G$ belongs to $\GG$,
then $\GG'$ is $(\PP,p,q)$-thin.
\end{lemma}
\begin{proof}
Let $t$ bound the treewidth of the overlays of $\GG$.
Let $r$ and $k$ be positive integers.  Let an $n$-vertex graph $G\in\GG'$ have components $G_1, \ldots, G_m\in \GG$,
and let $\LL_1,\ldots,\LL_m$ be the systems of $\PP_r$-overlays of $G_1$, \ldots, $G_m$
of thickness at most $1+1/k$ and
treewidth at most $t(r,3k)$, such that $|\LL_1|=\ldots=|\LL_m|\le \poly(q(n))$, found in
total time $O(p(n)+n\,\poly(q(n)))$ by Lemma~\ref{lemma-samesize}.
Then $\LL=\LL_1\circ\ldots\circ\LL_m$ is a system of $\PP_r$-overlays of $G$ of
thickness at most $1+1/k$ and treewidth at most $t(r,3k)$ with $|\LL|=|\LL_1|\le\poly(q(n))$.
\end{proof}

We will need the following fact.
\begin{observation}\label{obs-cutoff}
Let $x$, $y$, $m_x$, and $m_y$ be integers such that $|m_x-m_y|\le 1$.
If $y\ge x-1$, then $\min(y,m_y)\ge\min(x,m_x)-1$.
\end{observation}

We now consider graphs with layerings such that unions of any few layers belong to a thin class.
\begin{lemma}\label{lemma-layering}
Consider any $\PP\in\{\Aa,\Ss,\bigstar\}$.
For every positive integer $i$, let $\GG_i$ be a $(\PP,p,q)$-thin class of graphs.
Let $\GG$ be a class of graphs such that every $G\in\GG$ has a layering $(V_1,\ldots, V_d)$ such that
for $1\le i\le d$, the union of each at most $i$ consecutive layers induces a subgraph of $G$ belonging to $\GG_i$.
Suppose furthermore that such a layering can be found in time $O(p(n)+n\,\poly(q(n))$, where $n=|V(G)|$.
Then $\GG$ is a $(\PP,p,q)$-thin.
\end{lemma}
\begin{proof}
For a positive integer $i$, let $t_i$ denote the function bounding the treewidth of overlays of $\GG_i$.
Let $G\in\GG$ be an $n$-vertex graph, let $r$ and $k$ be positive integers, and let $(V_1,\ldots,V_d)$ be
the layering of $G$ from the assumptions of the lemma. For any integer $i$ such that $i\le 0$ or $i>d$, let $V_i=\emptyset$.

We need to find a system of $\PP_r$-overlays of $G$
of thickness at most $1+\frac{1}{k}$, size at most $\poly(q(n))$, and bounded treewidth, in time $O(p(n)+n\,\poly(q(n))$.
Let $\Delta=6kr$.  If $d\le\Delta$, then $G\in\GG_\Delta$ and such a system can be found using the algorithm
for the class $G_\Delta$. Hence, assume that $d\ge\Delta$.

For an integer $j$, let $G_j$ be the subgraph of $G$ induced by
$V_{j-r}\cup \ldots\cup V_{j+\Delta+r-1}$.  Note that $G_j\in \GG_{\Delta+2r}$.  Using Lemma~\ref{lemma-samesize},
for each integer $j$ such that $V(G_j)\neq\emptyset$ we obtain a system $\LL_j$ of $\PP_r$-overlays of $G_j$ of thickness at most $1+\frac{1}{2k}$ and
treewidth at most $t_{\Delta+2r}(r,6k)$, such that for all such integers $j$ the size of the system $\LL_j$ is the same (and bounded by $\poly(q(n))$).
Since each vertex of $G$ belongs to at most $\Delta+2r$ of the graphs $G_j$, all these systems of overlays
can be constructed in time $O(p(n)+n\,\poly(q(n)))$.

Consider an integer $j$ with $V(G_j)\neq\emptyset$, and an overlay $L=(H,f,\ell)\in\PP_r(G_j)$.  For $x\in V(H)$ such that $f(x)\in V_i$ for some integer $i$, let $m_L(x)=r$
if $j\le i\le j+\Delta-1$, $m_L(x)=r-(j-i)$ if $i<j$ and $m_L(x)=r-(i-(j+\Delta-1))$ if $i>j+\Delta-1$;
note that $m(x)\ge 0$.  In the case that $\PP=\bigstar$, for $x\in V(H)$ such that $f(x)\in V(G^\star_j)\setminus V(G_j)$, let $m_L(x)$ be the maximum of $\{m_L(y):xy\in E(H)\}$;
note that the neighborhood $N(x)$ of $x$ in $H$ induces a clique, and thus $f(N(x))$ is contained in at most two consecutive layers
of the layering, and hence $m_L(N(x))$ consists of at most two consecutive integers.
Let the system $\LL'_j$ be obtained from $\LL_j$ as follows: for each $L=(H,f,\ell)\in \LL_j$, let us replace $\ell$ by the function $\ell'$
defined by $\ell'(x)=\min(\ell(x),m_L(x))$ for all $x\in V(H)$.
Let $L'=(H,f,\ell')$.  Since $\ell'(x)=0$ for all $x$ such that $f(x)\in V_{j-r}\cup V_{j+\Delta-1+r}$
and $(V_1,\ldots,V_d)$ is a layering, Observation~\ref{obs-cutoff} implies that $L'$ taken as an overlay of $G$ (or $G^\star$) and not just of $G_j$ (or $G^\star_j$)
is walk-preserving.

Let us consider an integer $b$ such that $0\le b\le \Delta-1$, and for any integer $j$ such that $j\bmod \Delta=b$ and $V(G_j)\neq\emptyset$,
let $L'_j=(H_j,f_j,\ell'_j)$ be any element of $\LL'_j$, obtained according to the previous paragraph from $L_j=(H_j,f_j,\ell_j)$.  Let
$$L^\circ=\bigcirc_{\substack{j\bmod \Delta=b,\\V(G_j)\neq\emptyset}} L'_j.$$
We claim that $L^\circ=(H^\circ,f^\circ,\ell^\circ)$ belongs to $\PP_r(G)$.
Indeed, consider a vertex $v$ of the graph overlaid by $L^\circ$; we need to show that there exists a vertex $x\in V(H^\circ)$ such
that $f^\circ(x)=v$ and $\ell^\circ(x)=r$.
If $v\in V(G)$, then let $v'=v$; if $\PP=\bigstar$ and $v\in V(G^\star)\setminus V(G)$, then let $v'$ be an arbitrary neighbor of $v$ in $V(G)$,
contained in the clique represented by $v$.  Let $V_i$ be the layer containing $v'$.  There exists unique $j$ such that $j\le i\le j+\Delta-1$
and $j\bmod\Delta=b$.  Note that $v\in V(G^\star_j)$.  Since $L_j$ belongs to $\PP_r(G_j)$, there exists $x\in V(H_j)$ such that $f_j(x)=v$ and $\ell_j(x)=r$.
By the choice of $\ell'_j$, observe that $\ell^\circ(x)=\ell'_j(x)=\ell_j(x)=r$.  Hence, $L^\circ$ is an $r$-neighborhood overlay of $G$ if $\PP\neq\bigstar$ and of $G^\star$ if $\PP=\bigstar$.
In the case that $\PP=\bigstar$, furthermore note that since each overlay $L'_j$ is $(V(G^\star_j)\setminus V(G_j))$-simpliciality-preserving,
$L^\circ$ is $(V(G^\star)\setminus V(G))$-simpliciality-preserving.

For $0\le b\le \Delta-1$, let $$\LL^b=\bigcirc_{\substack{j\bmod \Delta=b,\\V(G_j)\neq\emptyset}} \LL'_j.$$
By the previous paragraph, $\LL^b$ is a system of $\PP_r$-overlays for $G$, and observe that the treewidth of $\LL^b$ is at most $t_{\Delta+2r}(r,6k)$.
Let $D_b=\{l:l\equiv b-r,b-r+1,\ldots, b+r-1\pmod\Delta\}$.  Observe that for any index $l\in D_b$ and a vertex $v\in V_l$,
there exist two indices $j$ such that $j\bmod \Delta=b$ and $v\in V(G_j)$, while for $l\not\in D_b$ and $v\in V_l$, there exists
only one such index.  Hence, $\theta_{\LL^b}(v)\le 2+\frac{1}{k}\le 1+\frac{1}{2k}+3/2$ in the former case and $\theta_{\LL^b}(v)\le 1+\frac{1}{2k}$ in the
latter case.  Similarly, if $\PP=\bigstar$, $v$ is a vertex of $V(G^\star)\setminus V(G)$, and $v$ has a neighbor in $V_l$ for some $l\not\in D_b$, then
$v\in V(G_j^\star)$ for a unique index $j$ such that $j\bmod \Delta=b$, and otherwise $v\in V(G_j^\star)$ for two such indices $j$,
with the same resulting bounds on the thickness of $v$.

Let $\LL=\bigcup_{b=0}^{\Delta-1} \LL^b$.  Note that each integer $l$ belongs to $D_b$ for exactly $2r$ values of $b$,
and thus for each vertex $v$ of the graph overlaid by $\LL$, we have
$$\theta_{\LL}(v)=1+\frac{1}{2k}+\frac{3}{2}\cdot\frac{2r}{\Delta}=1+1/k.$$
Therefore, $\LL$ is a system of $\PP_r$-overlays of $G$ of thickness at most $1+\frac{1}{k}$, size at most $\poly(q(n))$, and treewidth at most $t_{\Delta+2r}(r,6k)$,
as required.
\end{proof}

The next lemma deals with apex vertices; as such, it cannot apply to subgraph-based overlays by Lemma~\ref{lemma-localtw}, since adding a universal vertex
makes locally bounded treewidth impossible (unless the class has bounded treewidth).

\begin{lemma}\label{lemma-apex}
Consider any $\PP\in\{\Aa,\bigstar\}$.
Let $\GG$ be a $(\PP,p,q)$-thin class of graphs for some functions $p$ and $q$ and let $a$ be a positive integer.
Let $\GG'$ be a class of graphs.  Suppose that for each $G\in\GG'$, there exists a set $A\subseteq V(G)$ of size at most $a$
such that $G-A\in\GG$.  Suppose furthermore that such a set $A$ can be found in time $O(p(n)+n\,\poly(q(n)))$, where $n=|V(G)|$.
Then $\GG'$ is $(\PP,p,q)$-thin.
\end{lemma}
\begin{proof}
Let $t$ bound the treewidth of overlays of $\GG$.
Let $G$ be an $n$-vertex graph from $\GG'$, let $A\subseteq V(G)$ be a set of size at most $a$
such that $G-A\in\GG$, and let $r$ and $k$ be positive integers.
Let $\LL$ be a system of $\PP_r$-overlays of $G-A$ of thickness at most $1+1/k$, treewidth at most $t(r,k)$, and size
at most $\poly(q(n))$, found in time $O(n\,\poly(q(n)))$.  Let $A=\{v_1,\ldots,v_b\}$ for some $b\le a$.

Consider any overlay $L=(H,f,\ell)\in\LL$.
Let $H'$, $f'$, and $\ell'$ be obtained from $H$, $f$, and $\ell$ as follows: to $H$, we add vertices $x_1$, \ldots, $x_b$ with $\ell'(x_1)=\ldots=\ell'(x_b)=r$,
edges between vertices $x_i$ and $x_j$ for all $i,j\in\{1,\ldots, b\}$
such that $v_iv_j\in E(G)$, and edges $x_iy$ for all $i\in \{1,\ldots,b\}$ and $y\in V(H)$ such that $v_if(y)\in E(G)$, and we set $f'(x_i)=v_i$ for $1\le i\le b$.
Additionally, if $\PP=\bigstar$, we modify the resulting overlay as follows.  For each non-empty clique $K$ of $G[A]$, we add a vertex $x$ with $\ell'(x)=r$ adjacent to vertices of $(f')^{-1}(K)$ and let $f'(x_K)$ be the corresponding vertex of $G^\star$.
Also, for each clique $K$ of $G$ such that $K\cap A\neq\emptyset$ and $K\not\subseteq A$, letting $v$ be the vertex of $(G-A)^\star$ corresponding to $K\setminus A$ and $v'$ the vertex of $G^\star$ corresponding
to $K$, for each $x\in V(H)$ such that $f(x)=v$ we add a vertex $x'$ to $H'$ with $\ell'(x')=\ell(x)$ adjacent to the same vertices as $x$ in addition to the vertices of
$(f')^{-1}(K\cap A)$, and we set $f'(x')=v'$.  We define $L'=(H',f',\ell')$ and we observe that $L'\in\PP_r(G)$.

Let $\LL'=\{L':L\in\LL\}$.  Note that $\LL'$ is a system of $\PP_r$-overlays of $G$ of treewidth at most $t(r,k)+a$ and thickness $\theta(\LL')=\theta(\LL)$.
Hence, the claim of the lemma holds.
\end{proof}

Using the same argument, we can show the following.

\begin{corollary}\label{cor-root}
Consider any $\PP\in\{\Aa,\bigstar\}$.
Let $\GG$ be a $(\PP,p,q)$-thin class of graphs for some functions $p$ and $q$, with $t$ bounding the treewidth of its overlays.
Let $a$, $k$ and $r$ be positive integers, let $G$ be a graph belonging to $\GG$ with $n$ vertices,
and let $A$ be a set of its vertices of size at most $a$.  In time $O(p(n)+n\,\poly(q(n)))$, we can find a system $\LL$ of $\PP_r$-overlays of $G$
of thickness at most $1+1/k$, treewidth at most $t(r,k)+a$, and size at most $\poly(q(n))$, such that $\theta_\LL(v)=1$ for all $v\in A$.
\end{corollary}
\begin{proof}
Let $\LL'$ be a system of $\PP_r$-overlays of $G$ of thickness at most $1+1/k$,
treewidth at most $t(r,k)$, and size at most $\poly(q(n))$, found in time $O(p(n)+n\,\poly(q(n)))$.
By Observation~\ref{obs-sg}, we can transform $\LL'$ into a system $\LL''$ of $\PP_r$-overlays of $G-A$
of thickness at most $1+1/k$, treewidth at most $t(r,k)$, and size $|\LL'|$.  The construction from the proof of Lemma~\ref{lemma-apex}
applied to $\LL''$ gives the requested system $\LL$ of $\PP_r$-overlays of $G$.
\end{proof}

Let $G_0$, $G_1$, \ldots, $G_m$ be graphs such that for $1\le i\le m$, $G_0\cap G_i$ is a clique,
and for $1\le i<j\le m$, $G_i\cap G_j\subseteq G_0$.  We say that the graph $G_0\cup G_1\cup\ldots\cup G_m$
is a \emph{star sum of $G_0$, \ldots, $G_m$, with center $G_0$ and rays $G_1$, \ldots, $G_m$}.
The following lemma gives the main motivation for introducing the notion of $\bigstar$-thinness.

\begin{lemma}\label{lemma-one-level-sum}
For $i\in\{1,2\}$ and some functions $p$ and $q$, let $\GG_i$ be a $(\bigstar,p,q)$-thin class of graphs.
Let $\GG$ be the class of star sums with center from $\GG_1$ and rays from $\GG_2$.
Suppose furthermore that for each graph $G\in \GG$ with $n$ vertices, we can find the center and
rays of such a star sum in time $O(p(n)+n\,\poly(q(n)))$.
Then $\GG$ is $(\bigstar,p,q)$-thin.
\end{lemma}
\begin{proof}
Let $t$ bound the treewidth of overlays of both $\GG_1$ and $\GG_2$.
Let $G\in\GG$ be a star sum of $G_0\in\GG_1$ and $G_1,\ldots,G_m\in \GG_2$, and let $n$ denote the number of vertices of $G$.
For $i=1,\ldots, m$, let $A_i=V(G_0\cap G_i)$.
Let $a=4t(1,1)+1$.  Since $A_i$ induces a clique in $G$, Lemma~\ref{lemma-nume} implies $|A_i|\le a$.

Let $r$ and $k$ be positive integers. Let $\LL_0$ be a system of $\bigstar_r$-overlays of $G_0$ of thickness at most $1+\frac{1}{3k}$, treewidth at most $t(r,3k)$,
and size at most $\poly(q(n))$.  By Lemma~\ref{lemma-samesize} and Corollary~\ref{cor-root}, there exist systems $\LL_1$, \ldots, $\LL_m$ of $\bigstar_r$-overlays
of $G_1$, \ldots, $G_m$ all of the same size (at most $poly(q(n))$), thickness at most $1+\frac{1}{3k}$, and treewidth at most $t(r,9k)+a$,
such that for $i=1,\ldots,m$, each vertex $v\in A_i$ satisfies $\theta_{\LL_i}(v)=1$.
All these systems can be found in total time $O(p(n)+n\,\poly(q(n)))$.

Consider any overlays $L_i=(H_i,f_i,\ell_i)\in\LL_i$ for $i=0,\ldots, m$.  For $i\in\{1,\ldots,m\}$, note that $f_i^{-1}(A_i)$ is a clique in $H_i$ and $\ell_i(x)=r$ for all $x\in f_i^{-1}(A_i)$.
Let $R_i$ be the set of vertices of $V(G_i^\star)\setminus V(G_i)$ representing cliques that are subsets of $A_i$, and let $H'_i=H_i-f_i^{-1}(R_i)$.
Let $v_{A_i}$ be the vertex of $G^\star$ representing the clique $A_i$.
Let $H$ be the graph obtained from $H_0$ by, for $i=1,\ldots,m$ and each vertex $x\in f_0^{-1}(v_{A_i})$, adding a copy $H'_{i,x}$ of the graph $H'_i$ with the clique $f_i^{-1}(A_i)$ identified
with the clique on the neighbors of $x$ (so that the functions $f_0$ and $f_i$ match on the common clique).
Let $f:V(H)\to V(G^\star)$ match $f_0$ on $V(H_0)$ and $f_i$ on each copy $H'_{i,x}$ of $H'_i$ for $i=1,\ldots,m$.
Let $\ell:V(H)\to \{0,1,\ldots,r\}$ match $\ell_0$ on $V(H_0)$, and for $i=1,\ldots,m$, each vertex $x\in f_0^{-1}(v_{A_i})$, and all vertices $y\in V(H'_{i,x})$, set $\ell(y)=\min(\ell_i(y),\ell_0(x))$.
Let us define $L=L(L_0,\ldots, L_m)\colonequals (H,f,\ell)$.

Clearly, $L$ is a $\star$-overlay of $G$.  We claim that $L$ is walk-preserving.  Indeed, consider any vertex $y\in V(H)$ with $\ell(y)>0$ and a neighbor $v$ of $f(y)$ in $G^\star$.
If $v,f(y)\in V(G_0^\star)$, or if $f(y)\in V(G_i^\star)\setminus V(G_0^\star)$ (and consequently $v\in V(G_i^\star)$) for some $i\in\{1,\ldots, m\}$, then $y$ has a neighbor $z$ with $\ell(z)\ge \ell(y)-1$ and $f(z)=v$ by
the fact that $L_0$ or $L_i$ is walk-preserving and by Observation~\ref{obs-cutoff}.  Hence, suppose that $f(y)\in V(G_0^\star)$ and $v\in V(G_i^\star)\setminus V(G_0^\star)$
for some $i\in\{1,\ldots, m\}$, and thus $f(y)\in A_i$.  Since $L_0$ is walk-preserving and $v_{A_i}$ is a neighbor of $f(y)$ in $G_0^\star$, there exists a neighbor $x$ of $y$ in $H_0$
such that $f(x)=v_{A_i}$ and $\ell(x)\ge \ell(y)-1$.  Since $L_i$ is walk-preserving, there exists a neighbor $z$ of $y$ in $H'_{i,x}$ such that $f(z)=v$ and $\ell_i(z)\ge\ell_i(y)-1=r-1$.
Consequently, $\ell(z)=\min(\ell_i(z),\ell_0(x))\ge\ell(y)-1$.

Furthermore, since $L_0\in\bigstar_r(G_0)$, for $i=1,\ldots, m$, there exists $x\in f_0^{-1}(v_{A_i})$ such that $\ell_0(x)=r$.  Consequently, since $L_i\in \bigstar_r(G_i)$,
for any vertex $v\in V(G_i^\star)\setminus V(G_0^\star)$ there exists a vertex $y\in V(H'_{i,x})$ such that $f(y)=v$ and $\ell(y)=\min(\ell_i(y),\ell_0(x))=r$.
We conclude that $L\in\bigstar_r(G)$.

Let $\LL_0=\{L_0^1,\ldots,L_0^b\}$ and for $i=1,\ldots,m$, let $\LL_i=\{L_i^1,\ldots,L_i^c\}$. 
For $1\le j\le b$ and $1\le l\le c$, let $L_{j,l}=L(L_0^j,L_1^l,L_2^l,\ldots, L_m^l)$,
and let $\LL=\{L_{j,l}:1\le j\le b,1\le l\le c\}$.  Then $\LL$ is a system of $\bigstar_r$-overlays of $G$ of size $bc\le\poly(q(n))$ and treewidth at most $t(r,9k)+a$.
Furthermore, consider any vertex $v\in V(G^\star)$.  If $v\in V(G_0^\star)$, then $\theta_{\LL}(v)=\theta_{\LL_0}(v)<1+1/k$.
If $v\in V(G_i^\star)\setminus V(G_0^\star)$ for some $i\in \{1,\ldots, i\}$, then $\theta_{L_{j,l}}(v)=\theta_{L_0^j}(v_{A_i})\theta_{L_i^l}(v)$,
and thus
$$\theta_{\LL}(v)=\frac{1}{bc}\sum_{j,l}\theta_{L_0^j}(v_{A_i})\theta_{L_i^l}(v)=\theta_{\LL_0}(v_{A_i})\theta_{\LL_i}(v)\le \Bigl(1+\frac{1}{3k}\Bigr)^2<1+1/k.$$
Hence, $\LL$ has thickness at most $1+1/k$.
\end{proof}

Next, we consider a generalization of a star sum.
Let $G$ be a graph and let $(V_1,\ldots, V_d)$ is its layering. We say that $(V_1,\ldots, V_d)$ is a \emph{shadow-complete layering} if
for $i=1,\ldots, d-1$ and for each connected component $C$ of the graph $G[V_{i+1}\cup V_{i+2}\cup\ldots\cup V_d]$,
the neighbors of vertices of $C$ in $V_i$ induce a clique in $G$.  

\begin{lemma}\label{lemma-shadow}
Let $\GG$ be a $(\bigstar,p,q)$-thin class of graphs for some functions $p$ and $q$.
For a positive integer $d$, let $\GG_d$ be a class of graphs $G$ that have a shadow-complete layering with at most $d$
layers such that each layer induces a graph from $\GG$.  Furthermore, assume such a layering can be found in time $O(p(n)+n\,\poly(q(n)))$,
where $n=|V(G)|$.
Then $\GG_d$ is $(\bigstar,p,q)$-thin.
\end{lemma}
\begin{proof}
We prove the claim by induction on $d$.  The case $d=1$ is trivial, hence assume that $d\ge 2$.
Consider a graph $G\in\GG_d$ with a shadow-complete layering $(V_1,\ldots, V_d)$ such that each layer induces a subgraph belonging to $\GG$.
Let $G_0=G[V_1]$.  By Lemma~\ref{lemma-nume}, each clique in $G_0$ has size at most $a=4t(1,1)+1$,
where $t$ bounds the treewidth of overlays of $\GG$.

Let $G'=G[V_2\cup\ldots\cup V_d]$, and note that $(V_2,\ldots,V_d)$ is a shadow-complete layering of $G'$; hence,
$G'\in \GG_{d-1}$, which is a $(\bigstar,p,q)$-thin class by the induction hypothesis.  Let $G'_1$, \ldots, $G'_m$ be the
connected components of $G'$.  For $1\le i\le m$, let $A_i$ be the set of neighbors of $G'_i$ in $V_1$,
and let $G_i=G[V(G'_i)\cup A_i]$.  Note that $G$ is a star sum with center $G_0$ and rays $G_1$, \ldots, $G_m$.
By Lemmas~\ref{lemma-apex} and \ref{lemma-one-level-sum}, we conclude that $\GG_d$ is $(\bigstar,p,q)$-thin.
\end{proof}

Using Lemma~\ref{lemma-layering}, we can lift the previous result to shadow-complete layerings with an unbounded number of layers.

\begin{corollary}\label{cor-shadowmany}
Let $\GG$ be a $(\bigstar,p,q)$-thin class of graphs for some functions $p$ and $q$.
Let $\GG'$ be a class of graphs $G$ that have a shadow-complete layering such that each layer induces a graph from $\GG$.
Furthermore, assume such a layering can be found in time $O(p(n)+n\,\poly(q(n)))$, where $n=|V(G)|$.
Then $\GG'$ is $(\bigstar,p,q)$-thin.
\end{corollary}
\begin{proof}
Consider any graph $G\in \GG'$ and let $(V_1,\ldots, V_l)$ be its shadow-complete layering with each layer inducing a graph
from $\GG$.  Consider any positive integers $d$ and $i\le l-d+1$, and let $G_{i,d}=G[V_i\cup\ldots\cup V_{i+d-1}]$.
Observe that $(V_i,\ldots,V_{i+d-1})$ is a shadow-complete layering of $G_{i,d}$, and thus $G_{i,d}$ belongs
to the class $\GG_d$ from Lemma~\ref{lemma-shadow}.  Consequently, $\GG'$ is $(\bigstar,p,q)$-thin by Lemma~\ref{lemma-layering}.
\end{proof}

A \emph{tree decomposition} of a graph $G$ is a pair $(T,\beta)$, where $T$ is a tree and $\beta$ assigns to each vertex of $T$ a subset of
vertices of $G$, satisfying the following conditions:
\begin{itemize}
\item For each edge $xy\in E(G)$, there exists a vertex $v\in V(T)$ such that $\{u,v\}\subseteq \beta(v)$, and
\item for each vertex $x\in V(G)$, the set $\{v\in V(T):x\in\beta(v)\}$ induces a non-empty connected subtree in $T$.
\end{itemize}
The sets $\beta(v)$ for $v\in V(T)$ are called the \emph{bags} of the decomposition.
We say that the decomposition is \emph{chordal} if $\beta(v)\cap\beta(v')$ induces a clique in $G$ for every edge $vv'\in E(T)$.
A tree decomposition $(T',\beta')$ of a subgraph $G'$ of $G$ is \emph{contained} in the decomposition $(T,\beta)$ if for each vertex $v'\in V(T')$,
there exists a vertex $v\in V(T)$ such that $\beta(v')\subseteq\beta(v)$.  We need the following auxiliary result
relating chordal tree decompositions and shadow-complete layerings.

\begin{lemma}[Dujmovi\'{c}, Morin, and Wood~\cite{layers}]\label{lemma-redlay}
Let $G$ be a graph that has a chordal tree decomposition $(T,\beta)$ of adhesion at most $a$.
Then $G$ has a shadow-complete layering such that each layer induces a subgraph of $G$ with
a chordal tree decomposition of adhesion at most $a-1$ contained in $(T,\beta)$.
Furthermore, this layering and the decompositions of each layer can be obtained in linear time given $(T,\beta)$.
\end{lemma}

Together with Corollary~\ref{cor-shadowmany}, this enables us to lift $\bigstar$-thinness of a graph class $\GG$
to the class of graphs with chordal tree decompositions whose bags belong to $\GG$.

\begin{lemma}\label{lemma-cliquesum}
Let $\GG$ be a $(\bigstar,p,q)$-thin class of graphs closed on induced subgraphs, for some functions $p$ and $q$.
For a non-negative integer $a$, let $\GG_a$ be a class of graphs $G$ that have a chordal tree decomposition $(T,\beta)$ of adhesion at most $a$
such that $G[\beta(z)]\in\GG$ for all $z\in V(T)$.
Furthermore, assume such a decomposition can be found in time $O(p(n)+n\,\poly(q(n))$, where $n=|V(G)|$.
Then $\GG_a$ is $(\bigstar,p,q)$-thin.
\end{lemma}
\begin{proof}
We prove the claim by induction on $a$.  Consider any graph $G\in\GG$ and its tree-decomposition $(T,\beta)$ as
described in the statement of the lemma.

If $a=0$, then $G[\beta(z)]$ is a component of $G$ for each $z\in V(T)$, and the claim
follows from Lemma~\ref{lemma-comp}.  Hence, suppose that $a\ge 1$.  Let $(V_1,\ldots, V_l)$ be a shadow-complete
layering of $G$ obtained by Lemma~\ref{lemma-redlay}.  Consider any $i\in\{1,\ldots,l\}$, and let $(T_i,\beta_i)$
be a chordal tree decomposition of adhesion at most $a-1$ of $G[V_i]$ contained in $(T,\beta)$.
Since $\GG$ is closed on induced subgraphs, $G[\beta_i(z)]\in\GG$ for all $z\in V(T_i)$.  Hence $G[V_i]$ belongs to
$\GG_{a-1}$, which is $(\bigstar,p,q)$-thin by the induction hypothesis.  Consequently, $\GG_a$ is
$(\bigstar,p,q)$-thin by Corollary~\ref{cor-shadowmany}.
\end{proof}

Let us restate the previous result in the terms of arbitrary (not necessarily chordal) tree decompositions.
Let $(T,\beta)$ be a tree decomposition of a graph $G$.  The \emph{torso} of a vertex $z\in V(T)$ is obtained from $G[\beta(z)]$ by
adding cliques on $\beta(z)\cap\beta(z')$ for all $zz'\in E(T)$.

\begin{corollary}\label{cor-cliquesum}
Let $\GG$ be a $(\bigstar,p,q)$-thin class of graphs closed on induced subgraphs, for some functions $p$ and $q$.
Let $\GG'$ be a class of graphs $G$ that have a tree decomposition $(T,\beta)$ such that the torso of each vertex of $T$
belongs to $\GG$.  Furthermore, assume such a decomposition can be found in time $O(p(n)+n\,\poly(q(n))$, where $n=|V(G)|$.
Then $\GG'$ is $(\bigstar,p,q)$-thin.
\end{corollary}
\begin{proof}
Let $t$ bound the treewidth of overlays of $\GG$.  Let $G$ be a graph from $\GG'$ and let $(T,\beta)$ be its tree decomposition
whose torsos belong to $\GG$.
By Lemma~\ref{lemma-nume}, each graph $\GG$ has clique number at most $a=4t(1,1)+1$, and hence the decomposition $(T,\beta)$
has adhesion at most $a$.  Let $G'$ be the graph obtained from $G$ by adding cliques on sets $\beta(z)\cap \beta(z')$ for all $zz'\in E(T)$.
Then $(T,\beta)$ is a chordal tree decomposition of $G'$ and $G'[\beta(z)]\in\GG$ for each $z\in V(T)$.
Hence, $G'$ belongs to the $(\bigstar,p,q)$-thin class $\GG_a$ from Lemma~\ref{lemma-cliquesum}.  Observation~\ref{obs-sg} implies that
$\GG'$ is $(\bigstar,p,q)$-thin.
\end{proof}
Note that equivalently, $\GG'$ is a class of graphs obtained from those in $\GG$ by clique-sums.

\subsection{Proper minor-closed classes}

We are now ready to show that proper minor-closed classes are thin.
Each graph $G$ has a trivial $r$-neighborhood overlay $(G,\text{id},r)$.  Hence, we get the following.

\begin{observation}\label{obs-tw-thin}
Every class of graphs of bounded treewidth is $(\Ss,n,1)$-thin.
\end{observation}

Dujmovi\'{c}, Morin, and Wood~\cite{layers} introduced the following notion.
Let $c$ be a positive integer.  A graph $G$ has \emph{layered treewidth at most $c$} if $G$ has a tree decomposition $(T,\beta)$
and a layering $(V_1,\ldots, V_l)$ such that $|V_i\cap\beta(z)|\le c$ for $1\le i\le l$ and each vertex $z\in V(T)$.
Note that union of any $d$ layers induces a subgraph of $G$ of treewidth at most $cd$.  Observation~\ref{obs-tw-thin}
together with Lemma~\ref{lemma-layering} implies the following.

\begin{corollary}\label{cor-layertw}
Let $c$ be a positive integer and let $p$ be a positive non-decreasing polynomial.
Suppose that $\GG$ is a class of graphs of layered treewidth at most $c$.
If the corresponding tree decomposition and layering of an $n$-vertex graph from $\GG$ can be found in time $O(p(n))$,
then $\GG$ is $(\Ss,p,1)$-thin.
\end{corollary}

Let us state a corollary of Robertson-Seymour structure theorem, see Grohe~\cite{grohe2003local} and
Dujmovi\'{c} et al.~\cite{layers} for more details.
\begin{theorem}\label{thm-struct}
For every proper minor-closed class $\GG$, there exist positive integers $a$ and $c$ as follows.
Every graph $G\in\GG$ has a tree decomposition $(T,\beta)$ and a system $\{A_z\subseteq \beta(z):z\in V(T)\}$
of sets such that for every $z\in V(T)$, the set $A_z$ has size at most $a$ and letting $F$ denote the torso of $z$,
the graph $F-A_z$ has layered treewidth at most $c$.  Furthermore, the tree decomposition of $G$
and the system $\{A_z:z\in V(T)\}$ can be found in time $O(|V(G)|^3)$, and the decomposition and layering
of $F-A$ for each vertex can be found in linear time.
\end{theorem}

Corollary~\ref{cor-layertw}, Lemmas~\ref{lemma-sgbas-star} and \ref{lemma-apex}, and Corollary~\ref{cor-cliquesum} imply the following.
\begin{corollary}\label{cor-minor}
Every proper minor-closed class $\GG$ is $(\bigstar,n^3,1)$-thin.
\end{corollary}

\subsection{Strongly sublinear separators}

A \emph{separation} of a graph $G$ is a pair $(A,B)$ of edge-disjoint subgraphs of $G$ such that $A\cup B=G$,
and the \emph{size} of the separation is $|V(A)\cap V(B)|$.  Observe that $G$ has no edge with one end with $V(A)\setminus V(B)$
and the other end in $V(B)\setminus V(A)$, and thus the set $V(A)\cap V(B)$ separates $V(A)\setminus V(B)$ from
$V(B)\setminus V(A)$ in $G$.
A separation $(A,B)$ is \emph{balanced} if $|V(A)\setminus V(B)|\le 2|V(G)|/3$ and $|V(B)\setminus V(A)|\le 2|V(G)|/3$.
Note that $(G,G-E(G))$ is a balanced separation.
For a graph class $\GG$, let $s_\GG(n)$ denote the smallest nonnegative integer such that every graph in $\GG$ with at most $n$ vertices
has a balanced separation of size at most $s_\GG(n)$.  We say that $\GG$ has \emph{sublinear separators} if
$\lim_{n\to\infty} \frac{s_\GG(n)}{n}=0$, and that $\GG$ has \emph{strongly sublinear separators} if
there exist constants $c\ge 1$ and $0\le \delta<1$ such that $s_\GG(n)\le cn^\delta$ for every $n\ge 0$.

We now aim to show that subgraph-closed graph classes with strongly sublinear separators and bounded maximum degree
are $\Ss$-thin.  We start by obtaining layerings of such graphs from their tree decompositions.
We use the following fact about tree decompositions of graphs of bounded maximum degree.
\begin{lemma}[Dvo\v{r}\'ak~\cite{twd}]\label{lemma-few}
Let $G$ be a graph of maximum degree at most $\Delta$.  If $G$ has treewidth at most $t$, then $G$ has a rooted
tree decomposition $(T,\beta)$ with bags of size at most $12(t+1)(\Delta+1)$ such that every vertex of $T$ has at
most two sons and for each $v\in V(G)$, the subtree $T[\{u:v\in \beta(u)\}]$ has depth at most $1+4\log(\Delta+1)$.
\end{lemma}
The proof of this lemma gives an algorithm to find such a decomposition in time $O(t\Delta|V(G)|^2)$, assuming that a tree decomposition
of $G$ of width at most $t$ is given.

Let $\GG$ be a class of graphs closed on induced subgraphs.  If $\GG$ has strongly sublinear separators,
then each graph $G\in \GG$ has sublinear treewidth; a tree decomposition of sublinear width can be constructed
recursively as follows. Find a balanced separator $(A,B)$ of order at most $s_{\GG}(|V(G)|)$,
find tree decompositions of $A-V(B)$ and $B-V(A)$ inductively, and obtain a tree decomposition of $G$ by joining
the two by an arbitrary edge and adding $V(A\cap B)$ to all the bags.  Hence, we obtain the following.
\begin{observation}\label{obs-twss}
If $s$ is a positive integer and every induced subgraph of an $n$-vertex graph $G$ has a balanced separator of order at most $s$,
then $G$ has a tree decomposition of width at most $s\lceil \log_{3/2} n\rceil$.  Furthermore, if a balanced separator of order
at most $s$ in each induced subgraph $G'$ of $G$ can be found in time $O(p(|V(G')|))$ for a positive non-decreasing polynomial $p$ of degree greater than $1$, then such a
decomposition can be constructed in time $O(p(n))$.
\end{observation}
Dvo\v{r}\'ak and Norin~\cite{dvorak2016strongly} proved that every subgraph-closed class of graphs with strongly sublinear separators has polynomial
expansion.  Using algorithm of Plotkin et al.~\cite{plotkin}, a balanced separator of strongly sublinear size can be obtained in
(less than) quadratic time (see Dvo\v{r}\'ak and Norin~\cite[Corollary 2]{dvorak2016strongly} for details).
Hence, we obtain the following.
\begin{corollary}\label{cor-twss}
Let $\GG$ be a subgraph-closed class of graphs with strongly sublinear separators.  Then there exists a nonnegative real number $\delta<1$
such that for an $n$-vertex graph $G\in\GG$, a tree decomposition of width $O(n^\delta)$ can be found in time $O(n^2)$.
\end{corollary}

We now combine the results to obtain layerings where unions of a small number of consecutive layers
induce subgraphs whose all connected components are small.

\begin{corollary}\label{cor-dglayering}
Let $\GG$ be a subgraph-closed class of graphs with strongly sublinear separators and bounded maximum degree.
There exists a positive integer $c$ and a non-negative real number $\delta<1$ as follows.
Every $n$-vertex graph $G\in \GG$ has a layering such that the union of any $b$ consecutive
layers induces a subgraph whose components have size at most $c^bn^\delta$.  Furthermore, such a layering
can be found in time $O(n^3)$.
\end{corollary}
\begin{proof}
By Corollary~\ref{cor-twss} and Lemma~\ref{lemma-few}, there exist positive integers $a$ and $c'$ and a non-negative real number $\delta<1$
such that for any $n$-vertex graph $G\in\GG$, we can in time $O(n^3)$ find a tree decomposition $(T,\beta)$ of $G$ of width at most $c'n^\delta$,
where $T$ is a rooted tree in which each vertex has at most two sons and for each $v\in V(G)$, the subtree $T[\{u:v\in \beta(u)\}]$ has depth at most $a$.
Let $U_0=\emptyset$.
For $i\ge 1$, let $U_i$ be the union of $\beta(z)$ for all vertices $z\in V(T)$ whose distance from the root of $T$ is at least $(i-1)a$ and less than $ia$,
and let $V_i=U_i\setminus U_{i-1}$.  Let $d$ be maximum integer such that $V_d$ is non-empty.  Observe that $(V_1,\ldots,V_d)$ is a layering of $G$.

For $1\le i\le d$ and $1\le b\le d-i+1$, let $V_{i,b}=V_i\cup V_{i+1}\cup \ldots\cup V_{i+b-1}$ and let $\beta_{i,b}$ be the function defined by $\beta_{i,b}(z)=\beta(z)\cap V_{i,b}$.
Let $X_{i,b}$ be the set of vertices $z\in V(T)$ such that $\beta_{i,b}(z)\neq\emptyset$.  Note that $(T,\beta_{i,b})$ is a tree decomposition of $G[V_{i,b}]$,
and that each component of $T[X_{i,b}]$ is a rooted tree of depth at most $a(b+2)$.  Consequently, each component of $G[V_{i,b}]$ has size at most
$2^{a(b+2)}c'n^\delta$, and thus Corollary~\ref{cor-dglayering} holds with $c=c'2^{a+2}$.
\end{proof}

We now proceed similarly to Lemma~\ref{lemma-layering}.

\begin{lemma}\label{lemma-sglaycom}
Let $\GG$ be a subgraph-closed class of graphs with strongly sublinear separators and bounded maximum degree.
Let $c$ and $\delta$ be the corresponding constants from Corollary~\ref{cor-dglayering}.
Let $r$, $k$, and $t$ be positive integers, and let $\alpha$ be a positive real number.
Let $n_1=t$, $s_1=1$, and $\theta_1=1$, and for $i\ge 2$, let $s_i=\lceil 4\alpha^{-1}(1+\alpha)^irk\rceil$, $n_i=\Bigl(\frac{n_{i-1}}{c^{s_i+2r}}\Bigr)^{1/\delta}$, and $\theta_i=\theta_{i-1}+4r/s_i$.
For an $n$-vertex graph $G\in \GG$ and every positive integer $i$ such that each component of $G$ has at most $n_i$ vertices,
we can in time $O(n^3)$ find a system of $\Ss_r$-overlays of $G$ of thickness at most $\theta_i$,
treewidth at most $t$, and size exactly $s_1\cdot\ldots\cdot s_i$.  Furthermore, $\theta_i\le 1+1/k$.
\end{lemma}
\begin{proof}
We prove the claim by induction on $i$.  The case $i=1$ is trivial by Observation~\ref{obs-tw-thin}, hence assume that $s\ge 2$.
Consider an $n$-vertex graph $G\in \GG$ such that each component of $G$ has at most $n_i$ vertices.
Suppose that $G$ has $m\ge 2$ connected components and that we can find systems $\LL_1$, \ldots, $\LL_m$ of $\Ss_r$-overlays 
of the components of $G$ with the described properties; then $\bigcirc_{j=1}^m \LL_j$ is a system of $\Ss_r$-overlays of $G$ of thickness at most $\theta_i$,
treewidth at most $t$, and size $s_1\cdot\ldots\cdot s_i$.  Hence, we can without loss of generality assume that $G$
is connected, and consequently $n\le n_i$.

Let $(V_1,\ldots, V_d)$ be the layering of $G$ obtained by Corollary~\ref{cor-dglayering}.
For integers $j$ such that $j\le 0$ or $j>d$, let $V_j=\emptyset$.  For an integer $j$, let $G_j$ be the subgraph of $G$ induced by
$V_{j-r}\cup \ldots\cup V_{j+s_i+r-1}$.  Each component of $G_j$ has size at most $c^{s_i+2r}n_i^\delta=n_{i-1}$.
For each $j$ such that $V(G_j)\neq\emptyset$, let $\LL_j$ be the system of $\Ss_r$-overlays of $G_j$ of thickness at most $\theta_{i-1}$,
treewidth at most $t$, and size exactly $s_1\cdot\ldots\cdot s_{i-1}$, obtained by the induction hypothesis.
Let $\LL'_j$ be the system of overlays of $G$ obtained from $\LL_j$ as in the proof of Lemma~\ref{lemma-layering}.
For $0\le b\le s_i-1$, let $$\LL^b=\bigcirc_{\substack{j\bmod \Delta=b,\\V(G_j)\neq\emptyset}} \LL'_j,$$
and let $\LL=\bigcup_{b=0}^{s_i-1} \LL^b$.  As in the proof of Lemma~\ref{lemma-layering}, we argue that
$\LL$ is a system of $\Ss_r$-overlays of $G$ of thickness at most $\theta_{i-1}(1+2r/s_i)\le \theta_{i-1}+4r/s_i=\theta_i$,
treewidth at most $t$ and size $s_1\cdot\ldots\cdot s_i$.

The last inequality uses the fact that
\begin{align*}
\theta_i&\le 1+4r\sum_{j=1}^\infty \frac{1}{s_j}\le 1 + \frac{\alpha}{k} \sum_{j=1}^\infty (1+\alpha)^{-j}\\
&=1+\frac{\alpha}{k}\cdot\frac{1}{\alpha}=1+1/k\le 2.
\end{align*}
\end{proof}

The thinness of subgraph-closed class of graphs with strongly sublinear separators and bounded maximum degree
now follows simply by a careful choice of parameters.  Note that $e^{(\log\log n)^2}=O(e^{\varepsilon \log n})=O(n^\varepsilon)$
for every $\varepsilon>0$, and thus the function $(\log\log n)^2$ is subpolynomial.

\begin{corollary}\label{cor-sublin}
Every subgraph-closed class $\GG$ of graphs with strongly sublinear separators and bounded maximum degree is $(\Ss,n^3,e^{(\log\log n)^2})$-thin.
\end{corollary}
\begin{proof}
Let $c$ and $\delta$ be the corresponding constants from Corollary~\ref{cor-dglayering}, and let $\varepsilon = \min(1/\delta - 1,1/2)$.
Let $r$ and $k$ be positive integers, let $c'=c^{2r}$, $\alpha=\varepsilon(1-\varepsilon)/2$, and $a=c^{8\alpha^{-1}rk}$.
Let $t=t(r,k)=\max\left(3,\Bigl(c'a^{(1+\alpha)^2}\Bigr)^{2/\varepsilon}\right)$.
Let $s_i$ and $n_i$ for any positive integer $i$ be as in the statement of Lemma~\ref{lemma-sglaycom}.
Note that $s_i\le 8\alpha^{-1}rk(1+\alpha)^i$ for all $i\ge 1$.
We have
$$n_i\ge \Bigl(\frac{n_{i-1}}{c'a^{(1+\alpha)^i}}\Bigr)^{1+\varepsilon}.$$
By induction on $i$, we will show that $$c'a^{(1+\alpha)^{i+1}}\le n_i^{\varepsilon/2}$$
holds for every $i\ge 1$.  Indeed, the choice of $t=n_1$ ensures that the base case $i=1$ is true.
For $i\ge 2$, we have
\begin{align*}
n_i^{\varepsilon/2}&\ge \Bigl(\frac{n_{i-1}}{c'a^{(1+\alpha)^i}}\Bigr)^{(1+\varepsilon)\varepsilon/2}\ge \bigl(c'a^{(1+\alpha)^i}\bigr)^{(2/\varepsilon-1)(1+\varepsilon)\varepsilon/2}\\
&=\bigl(c'a^{(1+\alpha)^i}\bigr)^{(1-\varepsilon/2)(1+\varepsilon)}=\bigl(c'a^{(1+\alpha)^i}\bigr)^{1+\alpha}\\
&\ge c'a^{(1+\alpha)^{i+1}},
\end{align*}
where the second inequality holds by the induction hypothesis.  We conclude that for $i\ge 1$, we have
$$n_i\ge n_{i-1}^{1+\alpha}\ge \ldots\ge n_1^{(1+\alpha)^{i-1}}\ge 3^{(1+\alpha)^{i-1}}.$$

Consider an $n$-vertex graph $G\in \GG$, and let $i=\Bigl\lceil\frac{1}{\log(1+\alpha)}\log\log n\Bigr\rceil+1$,
so that $n_i\ge n$.  By Lemma~\ref{lemma-sglaycom}, we can in time $O(n^3)$ find a system of $\Ss_r$-overlays of $G$ of thickness at most $1+1/k$,
treewidth at most $t(r,k)$, and size at most $s_1\cdot\ldots\cdot s_i\le s_i^i=\poly(e^{(\log\log n)^2})$.
\end{proof}

\bibliographystyle{acm}
\bibliography{thin}

\begin{thebibliography}{10}

\bibitem{baker1994approximation}
{\sc Baker, B.}
\newblock Approximation algorithms for {NP}-complete problems on planar graphs.
\newblock {\em Journal of the ACM (JACM) 41}, 1 (1994), 153--180.

\bibitem{gridtw}
{\sc Berger, E., Dvo{\v{r}}{\'a}k, Z., and Norin, S.}
\newblock Treewidth of grid subsets.
\newblock {\em Combinatorica\/} (2017).
\newblock accepted.

\bibitem{cabello2015simple}
{\sc Cabello, S., and Gajser, D.}
\newblock Simple {PTAS}'s for families of graphs excluding a minor.
\newblock {\em Discrete Applied Mathematics 189\/} (2015), 41--48.

\bibitem{courcelle}
{\sc Courcelle, B.}
\newblock The monadic second-order logic of graphs. {I}. {R}ecognizable sets of
  finite graphs.
\newblock {\em Information and computation 85}, 1 (1990), 12--75.

\bibitem{dawar2006approximation}
{\sc Dawar, A., Grohe, M., Kreutzer, S., and Schweikardt, N.}
\newblock Approximation schemes for first-order definable optimisation
  problems.
\newblock In {\em 21st Annual IEEE Symposium on Logic in Computer Science
  (LICS'06)\/} (2006), IEEE, pp.~411--420.

\bibitem{contrpart}
{\sc Demaine, E., Hajiaghayi, M., and Kawarabayashi, K.}
\newblock Contraction decomposition in {H}-minor-free graphs and algorithmic
  applications.
\newblock In {\em Proceedings of the Forty-third Annual ACM Symposium on Theory
  of Computing\/} (2011), STOC '11, ACM, pp.~441--450.

\bibitem{demaine2004equivalence}
{\sc Demaine, E.~D., and Hajiaghayi, M.}
\newblock Equivalence of local treewidth and linear local treewidth and its
  algorithmic applications.
\newblock In {\em Proceedings of the fifteenth annual ACM-SIAM symposium on
  Discrete algorithms\/} (2004), Society for Industrial and Applied
  Mathematics, pp.~840--849.

\bibitem{demaine2005bidimensionality}
{\sc Demaine, E.~D., and Hajiaghayi, M.}
\newblock Bidimensionality: new connections between {FPT} algorithms and
  {PTAS}s.
\newblock In {\em Proceedings of the sixteenth annual ACM-SIAM symposium on
  Discrete algorithms\/} (2005), Society for Industrial and Applied
  Mathematics, pp.~590--601.

\bibitem{demaine2005algorithmic}
{\sc Demaine, E.~D., Hajiaghayi, M.~T., and Kawarabayashi, K.-i.}
\newblock Algorithmic graph minor theory: {D}ecomposition, approximation, and
  coloring.
\newblock In {\em 46th Annual IEEE Symposium on Foundations of Computer Science
  (FOCS'05)\/} (2005), IEEE, pp.~637--646.

\bibitem{devospart}
{\sc DeVos, M., Ding, G., Oporowski, B., Sanders, D., Reed, B., Seymour, P.,
  and Vertigan, D.}
\newblock Excluding any graph as a minor allows a low tree-width 2-coloring.
\newblock {\em J. Comb. Theory, Ser. B 91\/} (2004), 25--41.

\bibitem{layers}
{\sc Dujmovi{\'c}, V., Morin, P., and Wood, D.~R.}
\newblock Layered separators in minor-closed families with applications.
\newblock {\em arXiv e-prints 1306.1595\/} (2013).

\bibitem{twd}
{\sc Dvo{\v{r}}{\'a}k, Z.}
\newblock Sublinear separators, fragility and subexponential expansion.
\newblock {\em European Journal of Combinatorics 52, Part A\/} (2016),
  103--119.

\bibitem{dvorak2016strongly}
{\sc Dvo{\v{r}}{\'a}k, Z., and Norin, S.}
\newblock Strongly sublinear separators and polynomial expansion.
\newblock {\em SIAM Journal on Discrete Mathematics 30\/} (2016), 1095--1101.

\bibitem{eppstein00}
{\sc Eppstein, D.}
\newblock Diameter and treewidth in minor-closed graph families.
\newblock {\em Algorithmica 27\/} (2000), 275--291.

\bibitem{fomin2011bidimensionality}
{\sc Fomin, F.~V., Lokshtanov, D., Raman, V., and Saurabh, S.}
\newblock Bidimensionality and {EPTAS}.
\newblock In {\em Proceedings of the twenty-second annual ACM-SIAM symposium on
  Discrete Algorithms\/} (2011), SIAM, pp.~748--759.

\bibitem{grigoriev2007algorithms}
{\sc Grigoriev, A., and Bodlaender, H.~L.}
\newblock Algorithms for graphs embeddable with few crossings per edge.
\newblock {\em Algorithmica 49}, 1 (2007), 1--11.

\bibitem{grohe2003local}
{\sc Grohe, M.}
\newblock Local tree-width, excluded minors, and approximation algorithms.
\newblock {\em Combinatorica 23}, 4 (2003), 613--632.

\bibitem{grohe2014deciding}
{\sc Grohe, M., Kreutzer, S., and Siebertz, S.}
\newblock Deciding first-order properties of nowhere dense graphs.
\newblock In {\em Proceedings of the 46th Annual ACM Symposium on Theory of
  Computing\/} (2014), ACM, pp.~89--98.

\bibitem{har2015approximation}
{\sc Har-Peled, S., and Quanrud, K.}
\newblock Approximation algorithms for polynomial-expansion and low-density
  graphs.
\newblock In {\em Algorithms-ESA 2015}. Springer, 2015, pp.~717--728.

\bibitem{hunt1998nc}
{\sc Hunt, H.~B., Marathe, M.~V., Radhakrishnan, V., Ravi, S.~S., Rosenkrantz,
  D.~J., and Stearns, R.~E.}
\newblock {NC}-approximation schemes for {NP}- and {PSPACE}-hard problems for
  geometric graphs.
\newblock {\em Journal of Algorithms 26}, 2 (1998), 238--274.

\bibitem{grad2}
{\sc Ne{\v{s}}et\v{r}il, J., and {Ossona de Mendez}, P.}
\newblock Grad and classes with bounded expansion {II}. {A}lgorithmic aspects.
\newblock {\em European J. Combin. 29\/} (2008), 777--791.

\bibitem{papadimitriou1988optimization}
{\sc Papadimitriou, C., and Yannakakis, M.}
\newblock Optimization, approximation, and complexity classes.
\newblock In {\em Proceedings of the twentieth annual ACM symposium on Theory
  of computing\/} (1988), ACM, pp.~229--234.

\bibitem{plotkin}
{\sc Plotkin, S., Rao, S., and Smith, W.~D.}
\newblock Shallow excluded minors and improved graph decompositions.
\newblock In {\em Proceedings of the fifth annual ACM-SIAM symposium on
  Discrete algorithms\/} (1994), Society for Industrial and Applied
  Mathematics, pp.~462--470.

\end{thebibliography}

\end{document}